\documentclass[a4paper,11pt]{article}
 \usepackage{fullpage}
\pdfoutput=1

\usepackage[utf8]{inputenc}
\usepackage[english]{babel}
\usepackage[T1]{fontenc}

\usepackage{csquotes}

\usepackage[
    backend=biber,
    language=auto,
    style = alphabetic,
    doi=false,
    isbn=false,
    url=false,
    maxnames=4,
    natbib=true,
    maxbibnames=15
    ]{biblatex}

\DeclareBibliographyDriver{unpublished}{%
  \usebibmacro{bibindex}%
  \usebibmacro{begentry}%
  \usebibmacro{author}%
  \setunit{\printdelim{nametitledelim}}\newblock
  \usebibmacro{title}%
  \newunit
  \printlist{language}%
  \newunit\newblock
  \usebibmacro{byauthor}%
  \newunit\newblock
  \printfield{howpublished}%
  \newunit\newblock
  \printfield{type}%
  \newunit\newblock
  \usebibmacro{eprint}%
  \newunit\newblock
  \usebibmacro{event+venue+date}%
  \newunit\newblock
  \printfield{note}%
  \newunit\newblock
  \usebibmacro{location+date}%
  \newunit\newblock
  \usebibmacro{addendum+pubstate}%
  \setunit{\bibpagerefpunct}\newblock
  \usebibmacro{pageref}%
  \newunit\newblock
  \iftoggle{bbx:related}
    {\usebibmacro{related:init}%
     \usebibmacro{related}}
    {}%
  \usebibmacro{finentry}}

\DeclareFieldFormat{sentencecase}{\MakeSentenceCase{#1}}
\renewbibmacro*{title}{%
  \ifthenelse{\iffieldundef{title}\AND\iffieldundef{subtitle}}
    {}
    {\ifthenelse{\ifentrytype{article}\OR\ifentrytype{inbook}%
      \OR\ifentrytype{incollection}\OR\ifentrytype{inproceedings}%
      \OR\ifentrytype{inreference}\OR\ifentrytype{unpublished}}
      {\printtext[title]{%
        \printfield[sentencecase]{title}%
        \setunit{\subtitlepunct}%
        \printfield[sentencecase]{subtitle}}}%
      {\printtext[title]{%
        \printfield[titlecase]{title}%
        \setunit{\subtitlepunct}%
        \printfield[titlecase]{subtitle}}}%
     \newunit}%
  \printfield{titleaddon}
}
\newbibmacro{string+doiurlisbn}[1]{%
  \iffieldundef{doi}{%
      \iffieldundef{url}{%
        \iffieldundef{isbn}{%
          \iffieldundef{issn}{%
            #1%
          }{%
            \href{http://books.google.com/books?vid=ISSN\thefield{issn}}{#1}%
          }%
        }{%
          \href{http://books.google.com/books?vid=ISBN\thefield{isbn}}{#1}%
        }%
      }{%
        \href{\thefield{url}}{#1}%
      }%
  }{%
    \href{http://dx.doi.org/\thefield{doi}}{#1}%
  }%
}

\newbibmacro{string+arxiv}[1]{%
        \href{https://arxiv.org/abs/\thefield{eprint}}{#1}%
}

\DeclareFieldFormat[unpublished]{title}{%
\usebibmacro{string+arxiv}{\mkbibquote{#1}}
}

\DeclareFieldFormat[article,incollection,inproceedings]{title}%
    {\usebibmacro{string+doiurlisbn}{\mkbibquote{#1}}}

\usepackage{amsmath}
\usepackage{amsfonts}
\usepackage{amsthm}
\usepackage{amssymb}
\usepackage{dsfont}
\usepackage{mathtools}
\usepackage{paralist}
\usepackage{thm-restate}

\usepackage[printonlyused,withpage,nolist]{acronym}
\makeatletter
\AtBeginDocument{%
  \renewcommand*{\AC@hyperlink}[2]{%
    \begingroup
      \hypersetup{hidelinks}%
      \hyperlink{#1}{#2}%
    \endgroup
  }%
}
\makeatother

\usepackage[table,usenames,dvipsnames]{xcolor}

\usepackage[colorlinks=true,
            hyperindex,
            linkcolor = NavyBlue,
            anchorcolor = hyperrefblue,
            citecolor = Green,
            filecolor = hyperrefblue,
            urlcolor = NavyBlue,
            breaklinks=true]{hyperref}

\usepackage[capitalise]{cleveref}
\usepackage{comment}
\crefname{equation}{Equation}{Equations}
\Crefname{equation}{Equation}{Equations}

\usepackage{complexity}

\newcommand{\NQMSA}{N_{\class{QMSA}}}
\newcommand{\length}[1]{\left|#1\right|}
\renewcommand{\H}{\mathcal{H}}

\renewcommand{\i}{\ensuremath\mathrm{i}}

\newcommand{\QMMWW}{\textrm{QUANTUM~MONOTONE~MINIMUM~SATISFYING~ASSIGNMENT}}
\newcommand{\gsb}{\ket{\mathrm{gs}_b}}

\newcommand{\Span}{\textup{Span}}

\newcommand{\norm}[1]{\left\lVert #1 \right\rVert}       
\newcommand{\enorm}[1]{\norm{#1}_{\mathrm{2}}}      
\newcommand{\snorm}[1]{\norm{#1}_{\mathrm {\infty}}}    

\newcommand{\psiyts}{\ket{\Psi_{y,t,s}}}

\let\Re\relax
\DeclareMathOperator{\Re}{Re}
\DeclareMathOperator{\tr}{tr}

\providecommand\given{}
\newcommand\SetSymbol[1][]{%
  \nonscript\:#1\vert
  \allowbreak
  \nonscript\:
  \mathopen{}}
\DeclarePairedDelimiterX\set[1]\{\}{%
  \renewcommand\given{\SetSymbol[\delimsize]}
  #1
}

\newcommand{\abs}[1]{\left\lvert #1 \right\rvert}

\newcommand{\complex}{{\mathbb C}}
\newcommand{\reals}{{\mathbb R}}

\newcommand{\nats}{{\mathbb N}}
\newcommand{\wt}[1]{\widetilde{#1}}
\newcommand{\wtt}{\wt{t}}

\newcommand{\wts}{\wt{s}}
\newcommand{\wto}{\wt{1}}

\newcommand{\RR}{\mathbb{R}}

\newcommand{\m}{m_V}
\newcommand{\n}{n_V}
\newcommand{\absA}{\abs{A}}
\newcommand{\absB}{\abs{B}}
\newcommand{\absC}{\abs{C}}
\newcommand{\absD}{\abs{D}}

\def\ket#1{ | #1 \rangle}
\def\bra#1{{\langle #1 | }}
\newcommand{\ketbra}[2]{\ket{#1}\!\bra{#2}}        
\newcommand{\braket}[2]{\mbox{$\langle #1  | #2 \rangle$}}

\newcommand{\sandwichb}[3]
  {\bigl\langle  #1 \bigr| #2 \bigl| #3 \bigr\rangle}

\newcommand{\SFGPQ}{S_{FGPQ}}
\newcommand{\etayst}{\ket{\eta_{y,s,t}}}
\newcommand{\etaoo}{\ket{\eta_{0,1,1}}}

\newcommand{\Tc}{\mathcal{T}_C}
\newcommand{\Td}{\mathcal{T}_D}
\newcommand{\VQA}{\textup{MIN-VQA}}
\newcommand{\QAOA}{\textup{MIN-QAOA}}
\newcommand{\opt}{\mathrm{opt}}

\newtheorem{theorem}{Theorem}

\newtheorem{lemma}{Lemma}

\newtheorem{cor}[theorem]{Corollary}
\newtheorem{obs}[theorem]{Observation}
\newtheorem{ass}[theorem]{Assumption}
\newtheorem{problem}{Problem}

\newtheorem{definition}{Definition}

\begin{document}
\date{}
\title{The optimal depth of variational quantum algorithms is QCMA-hard to approximate}

\author{Lennart Bittel\thanks{Institute for Theoretical Physics,
  Heinrich Heine University D{\"u}sseldorf, Germany. Email: lennart.bittel@uni-duesseldorf.de.} \and Sevag Gharibian\thanks{Department of Computer Science, and Institute for Photonic Quantum Systems, Paderborn University, Germany. Email: sevag.gharibian@upb.de.} \and Martin Kliesch\thanks{Institute for Theoretical Physics,
  Heinrich Heine University D{\"u}sseldorf, and Institute for Quantum-Inspired and Quantum Optimization, Hamburg University of Technology, Germany. Email: martin.kliesch@tuhh.de.}}

 \hypersetup{
       pdftitle = {The optimal depth of variational quantum algorithms is QCMA-hard to approximate},
       pdfauthor = {Lennart Bittel, Sevag Gharibian, Martin Kliesch},
       pdfsubject = {Quantum physics},
       pdfkeywords = {hybrid, quantum, algorithm,
              VQE, variational quantum eigensolver,
              VQA, variational quantum algorithms,
              QAOA, quantum approximate optimization algorithm,
              variational, circuit, depth,
              PQC, parameterized,
              barren, plateau,
              Hamiltonian, generator,
              optimization,
              QCMQ, QMA, NP, hard, complete, approximation,
              computational, complexity, theory
              }
      }

\maketitle

\begin{abstract}
    Variational Quantum Algorithms (VQAs), such as the Quantum Approximate Optimization Algorithm (QAOA) of [Farhi, Goldstone, Gutmann, 2014], have seen intense study towards near-term applications on quantum hardware. A crucial parameter for VQAs is the \emph{depth} of the variational ``ansatz'' used --- the smaller the depth, the more amenable the ansatz is to near-term quantum hardware in that it gives the circuit a chance to be fully executed before the system decoheres. In this work, we show that approximating the optimal depth for a given VQA ansatz is intractable. Formally, we show that for any constant $\epsilon>0$, it is QCMA-hard to approximate the optimal depth of a VQA ansatz within multiplicative factor $N^{1-\epsilon}$, for $N$ denoting the encoding size of the VQA instance. (Here, Quantum Classical Merlin-Arthur (QCMA) is a quantum generalization of NP.) We then show that this hardness persists
    in the even ``simpler'' QAOA-type settings. 
    To our knowledge, this yields the first natural QCMA-hard-to-approximate problems. 
\end{abstract}

\section{Introduction}

In the current era of Noisy Intermediate Scale Quantum (NISQ) devices, quantum hardware is (as the name suggests) limited in size and ability.
Thus, NISQ-era quantum algorithm design has largely focused on \emph{hybrid} classical-quantum setups, which ask: What types of computational problems can a classical supercomputer, paired with a \emph{low-depth} quantum computer, solve?
This approach, typically called Variational Quantum Algorithms (VQA), has been studied intensively in recent years (see, e.g. \cite{Cerezo20VariationalQuantumAlgorithms,Bharti21NoisyIntermediateScale} for reviews), with Farhi, Goldstone and Gutmann's \acf{QAOA} being a prominent example~\cite{FarGolGut14}.

More formally, VQAs roughly work as follows.
One first chooses a variational ansatz (i.e.\ parameterization) over a family of quantum circuits.
Then, one iterates the following two steps until a ``suitably good'' parameter setting is found:
\begin{enumerate}
    \item Use a classical computer to optimize the ansatz parameters variationally\footnote{In practice, this typically means heuristic optimization.}.
    \item Run the resulting parameterized quantum algorithm on a NISQ device to evaluate the ``quality'' of the chosen parameters (relative to the computational problem of interest).
\end{enumerate}

The essential advantage of this setup over more traditional quantum algorithm design techniques (such as full Trotterization of a desired Hamiltonian evolution) is that one can attempt to minimize the \emph{depth} of the ansatz used.
(A formal definition of ``depth'' is given in \Cref{def:VQA}; briefly, it is the number of Hamiltonian evolutions the ansatz utilizes.)
This possibility gives VQAs a potentially crucial advantage on near-term quantum hardware (i.e.\ noisy hardware without quantum error correction),
because a NISQ device can, in principle, execute a low-depth ansatz before the system decoheres,
i.e.\ before environmental noise destroys the ``quantumness'' of the computation. From an analytic perspective, low-depth ansatzes also have an important secondary benefit --- VQAs of superlogarithmic depth are exceedingly difficult to analyze via worst-case complexity. Sufficiently low-depth setups, however, sometimes \emph{can} be rigorously analyzed, with the groundbreaking QAOA work of~\cite{FarGolGut14} for MAX-CUT being a well-known example. Thus, estimating the optimal depth for a \ac{VQA} appears central to its use in near-term applications.

\subsection{Our results} In this work, we show that it is intractable to approximate the optimal depth for a given \ac{VQA} ansatz, even within large multiplicative factors.
Moreover, this hardness also holds for the restricted ``simpler'' case of the \ac{QAOA}.
To make our claim rigorous, we first define the \ac{VQA} optimization problem we study. (Intuition to follow.)

\begin{restatable}[VQA minimization (\VQA$(k,l)$)]{problem}{defvqa}\label{def:VQA}
    For an $n$-qubit system:
    \begin{itemize}
        \item Input:
        \begin{enumerate}
          \item Set $H=\set{H_i}$ of Hamiltonians\footnote{An $n$-qubit Hamiltonian $H$ is a $2^n\times 2^n$ Hermitian matrix. Any unitary operation $U$ on a quantum computer can be generated via an appropriate choice of Hamiltonian $H$ and evolution time $t\geq 0$, i.e. $U=e^{iHt}$.}, where $H_i$ acts non-trivially only on a subset\footnote{For \Cref{thm:QCMA}, it will suffice to take $k\in O(1)$. In principle, however, containment in QCMA holds for any $k\leq n$, so long as the $H_i$ are sparse in the standard Hamiltonian simulation sense~\cite{AT03}. By sparse, one means that each row $r$ of $H_i$ contains at most $r$ non-zero entries, which can be computed in poly-time given $r$.} $S_i\subseteq[n]$ of size $\abs{S_i}=k$.
          \item An $l$-local observable $M$ acting on a subset of $l$ qubits.
          \item Integers $0\leq m\leq m'$ representing circuit depth thresholds.
        \end{enumerate}
        \item Output:
        \begin{enumerate}
            \item YES if there exists a list of at most $m$ angles\footnote{Throughout \Cref{def:VQA}, for clarity we assume all angles are specified to $\poly(n)$ bits.
            } $(\theta_1,\ldots, \theta_m)\in \RR^m$ and a list $(G_1,\ldots, G_m)$
            of Hamiltonians from $H$ (repetitions permitted) such that
                \begin{align}\label{eqn:first}
                    \ket{\psi}\coloneqq e^{i\theta_mG_m}\cdots e^{i\theta_1G_1}\ket{0\cdots 0}
                \end{align}
                satisfies $\bra{\psi}M\ket{\psi}\leq 1/3$.

            \item NO if for all lists of at most $m'$ angles $(\theta_1,\dots,\theta_{m'})\in \RR^{m'}$ and all lists $(G_1,\ldots, G_{m'})$
            of Hamiltonians from $H$ (repetitions permitted),
                \begin{align}
                    \ket{\psi}\coloneqq e^{i\theta_{m'}G_{m'}}\cdots e^{i\theta_1G_1}\ket{0\cdots 0}
                \end{align}
                satisfies $\bra{\psi}M\ket{\psi}\geq 2/3$.
        \end{enumerate}
    \end{itemize}
\end{restatable}

\noindent For intuition, recall that a VQA ansatz is a parameterization over a family of quantum circuits. Above, the ansatz is parameterized by angles $\theta_j$, and the family of quantum circuits is generated by Hamiltonians $H_j$. The aim is to pick a \emph{minimum-length} sequence of Hamiltonian evolutions $e^{i\theta_jG_j}$, so that the generated state $\ket{\psi}$ has (say) low overlap with the target observable, $M$. For clarity, throughout this work, by ``depth'' of a VQA ansatz, we are referring to the standard VQA notion of the number of Hamiltonian evolutions $m$ applied\footnote{Alternatively, one could consider the \emph{circuit depth} of any simulation of the desired Hamiltonian sequence in \Cref{def:VQA}. The downside of this is that it would be much more difficult to analyze --- one would presumably first need to convert each $e^{i\theta_jG_j}$ to a circuit $U_j$ via a fixed choice of Hamiltonian simulation algorithm. One would then need to characterize the depth of the concatenated circuit $U_m\cdots U_1$.}.
(In the setting of QAOA, the ``depth'' is often referred to as the ``level'', up to a factor of $2$.)

We remark for \Cref{def:VQA} that we do not restrict the order in which Hamiltonians $H_i$ are applied, and any $H_i$ may be applied multiple times. Moreover, our results also hold if one defines the YES case to {maximize} overlap with $M$ (as opposed to minimize overlap).

Our first result is the following.

\begin{restatable}{theorem}{mainvqa}\label{thm:QCMA}
    \VQA$(k,l)$ is \QCMA-complete for $k\geq 4$, $l=2$, and $m\leq\poly(n)$.
    Moreover, for any $\epsilon>0$,
    it is \QCMA-hard to distinguish between the YES and NO cases of \VQA\ even if $m'/m\geq N^{1-\epsilon}$,
    where $N$ is the encoding size of the instance.
\end{restatable}

\noindent Here, Quantum-Classical Merlin-Arthur (QCMA) is a quantum generalization of NP with a classical proof and quantum verifier (formal definition in~\Cref{def:QCMA}).
For clarity, the \emph{encoding size} of the instance is the number of bits required to write down a \VQA\ instance, i.e. to encode $H=\set{H_i}$, $M$, $m$, $m'$ (see \Cref{def:VQA}). Note the encoding size is typically dominated by the encoding size of $H$, which may be assumed to scale as $\abs{H}$, i.e. with the number\footnote{Indeed, in the construction in the proof of \cref{thm:QCMA}, $N\in O(\abs{H})$.} of \emph{interaction terms} $H_i$, which can be asymptotically larger than the number of qubits, $n$. Thus, simple gap amplification strategies such as taking many parallel copies of all interaction terms do \emph{not} suffice to achieve our hardness ratio of $N^{1-\epsilon}$.

A direct consequence of \Cref{thm:QCMA} is that it is intractable (modulo the standard conjecture that $\BQP\neq\QCMA$, which also implies $\P\neq\QCMA$) to compute the optimum circuit depth within relative precision $N^{1-\epsilon}$ (proof given in \Cref{app:A} for completeness):

\begin{cor}[Depth minimization]\label{cor:consequence}
In \cref{def:VQA}, let $m_{\opt}$ denote the minimum depth $m$ such that $\bra{\psi}M\ket{\psi}\leq 1/3$. Then, for any constant $\epsilon>0$, computing estimate $m_{\textup{est}} \in [m_{\opt}, N^{1-\epsilon}m_{\opt}]$ is \QCMA-hard.
\end{cor}

On the other hand, even if a desired depth $m=m'$ is specified in advance,
it is also \QCMA-hard to find the minimizing angle and Hamiltonian sequences $(\theta_1,\ldots, \theta_m)$ and $(G_1,\ldots, G_m)$, respectively, which follows directly from \cref{thm:QCMA}:

\begin{cor}[Parameter optimization]
Consider \cref{def:VQA} with input $m=m'$.
Then the problem of finding the angles $(\theta_1, \dots, \theta_m)$ that minimize the expectation value $\bra{\psi}M\ket{\psi}$ is \QCMA-hard.
\end{cor}

We next turn to the special case of \acp{QAOA}. As detailed shortly under ``Previous work'', the study of QAOA ansatzes was initiated by~\cite{FarGolGut14} in the context of \emph{quantum} approximation algorithms for MAX CUT. In that work, a QAOA is analogous to a VQA, except there are only \emph{two} Hamiltonians $H=\set{H_b,H_c}$ given as input and $M$ is one of those two observables (see \Cref{def:QAOA} for a formal definition). For clarity, here we work with a more general definition of QAOA than~\cite{FarGolGut14}, in which neither $H_b$ nor $H_c$ need be diagonal in the standard basis. (In this sense, our definition is closer to the more general Quantum Alternating Operator Ansatz, also with acronym QAOA~\cite{HWORVB19}.)
For our hardness results, it will suffice for $H_b$ and $H_c$ to be \emph{$k$-local} Hamiltonians\footnote{A $k$-local $n$-qubit Hamiltonian $H$ is a quantum analogue of a MAX-$k$-SAT instance, and can be written $H=\sum_iH_i$, with each ``quantum clause'' $H_i$ acting non-trivially on some subset of $k$ qubits. Strictly speaking, each $H_i$ is tensored with the identity matrix on $n-k$ qubits to ensure all operators in the sum have the correct dimension.}.
 For QAOA, we show a matching hardness result:

\begin{restatable}{theorem}{thmqaoa}\label{thm:QCMA_QAOA}
    \QAOA$(k)$ is \QCMA-complete for $k\geq 4$ and $m\leq\poly(n)$.
  Moreover, for any $\epsilon>0$, it is \QCMA-hard to distinguish between the YES and NO cases of \QAOA\ even if $m'/m\geq N^{1-\epsilon}$,
   where $N$ is the number of strictly $k$-local terms comprising $H_b$ and $H_c$.
\end{restatable}

\noindent Note that in contrast to \VQA, which is parameterized by $k$ (the Hamiltonians' locality) and $l$ (the observable's locality), \QAOA\ is only parameterized by $k$. This is because in QAOA, the ``cost'' Hamiltonian $H_c$ \emph{itself} acts as the observable (in addition to helping drive the computation), which will be one of the obstacles we will need to overcome. For context, typically in applications of QAOA, $H_c$ encodes (for example~\cite{FarGolGut14}) a MAX CUT instance.

To the best of our knowledge, \Cref{thm:QCMA} and \Cref{thm:QCMA_QAOA} yield the first natural \QCMA-hard to approximate problems.

\subsection{Previous work}
Generally speaking, it is well-known that VQA parameters are ``hard to optimize'', both numerically and from a theoretical perspective. We now discuss selected works from the (vast) VQA literature, and clarify how these differ from our work.\\

\noindent\emph{1.\ Theoretical studies.} As previously mentioned, in 2014, Farhi, Goldstone and Gutmann proposed the \acf{QAOA}, a special case of VQA with only two local Hamiltonians $H=\set{H_b,H_c}$ (acting on $n$ qubits each). They showed that level-$1$ of the \ac{QAOA} (what we call ``depth 2'' in \Cref{def:VQA}) achieves a $0.6924$-factor approximation for the NP-complete MAX CUT problem. Unfortunately, worst-case analysis of higher levels has in general proven difficult, but Bravyi, Kliesch, Koenig and Tang~\cite{BKKT20} have shown an interesting negative result --- \ac{QAOA} to any \emph{constant} level/depth cannot outperform the classical Goemans-Williams algorithm for MAX CUT~\cite{GW95}. Thus, superconstant depth is \emph{necessary} if \ac{QAOA} is to have a hope of outperforming the best classical algorithms for MAX CUT. In terms of complexity theoretic hardness, Farhi and Harrow~\cite{Farhi2016QuantumSupremacy} showed that even level-$1$ QAOA's output distribution cannot be efficiently simulated by a classical computer.

Most relevant to this paper, however, is the work of Bittel and Kliesch~\cite{Bittel21TrainingVariationalQuantum}, which roughly shows that finding the optimal set of rotation angles (the $\theta_j$ in \Cref{def:VQA} and \Cref{def:QAOA}) is NP-hard. Let us clearly state how the present work differs from~\cite{Bittel21TrainingVariationalQuantum}:
\begin{enumerate}
     \item \cite{Bittel21TrainingVariationalQuantum} fixes both the depth of the VQA and the precise sequence of Hamiltonians $H_i$ to be applied as part of the input. It then asks: What is the complexity of computing the optimal rotation angles $\theta_i$ so as to minimize overlap with a given observable?

         In contrast, our aim here is to study the complexity of optimizing the \emph{depth} itself. Thus, \Cref{def:VQA} does not fix the depth $m$, nor the order/multiplicity of application of any of the Hamiltonian terms.

     \item \cite{Bittel21TrainingVariationalQuantum} shows that optimizing the rotation angles in \ac{QAOA} is \NP-hard, \emph{even if} one is allowed to work in time polynomial in the \emph{dimension} of the system. (Formally, this is obtained by reducing a MAX CUT instance of encoding size $N$ to \ac{QAOA} acting on $\log(N)$ qubits.)

         In contrast, we work in the standard setting of allowing only poly-time computations in the number of qubits, $n$, not the dimension. In return, we obtain stronger hardness results, both in that $\NP\subseteq\QCMA$ (and thus $\QCMA$-hardness is a stronger statement than NP-hardness\footnote{Note that for $\log(N)$-size instances of \ac{QAOA} as in \cite{Bittel21TrainingVariationalQuantum}, one cannot hope for more than NP-hardness, since both Hamiltonians $H_b$ and $H_c$ have polynomial dimension, and thus can be classically simulated efficiently. Thus, such instances are verifiable in NP.}), and in that we show hardness of approximation up to any multiplicative factor $N^{1-\epsilon}$.
\end{enumerate}

\noindent \emph{2.\ Practical/numerical studies.} For clarity, numerical studies are not directly related to our work. However, due to the intense practical interest in VQA for the NISQ era, for completeness we next survey some of the difficulties encountered when optimizing VQAs on the numerical side. For this, note that \acp{VQA} are typically used to solve problems which can be phrased as energy optimization problems (such as NP-complete problems like MAX CUT~\cite{FarGolGut14}).

In this direction, two crucial problems can arise in the classical optimization part of the standard VQA setup:
(i) barren plateaus \cite{McClean2018BarrenPlateausIn}, 
which lead to vanishing gradients, and (ii) local minima \cite{Bittel21TrainingVariationalQuantum}, many of which can be highly non-optimal.
Such unwanted local minima are also called \emph{traps}.
In order to counterbalance these challenges, heuristic optimization strategies have led to promising results in relevant cases but with not too many qubits.
Initialization-dependent barren plateaus \cite{McClean2018BarrenPlateausIn} can be avoided by tailored initialization \cite{Zhou2018QuantumApproximate}, and there are indications that barren plateaus are a less significant challenge than traps \cite{Anschuetz22BeyondBarrenPlateaus}.
In general, the optimization can be improved using
    natural gradients \cite{Wierichs2020AvoidingLocalMinima},
    multitask learning type approach \cite{Zhang20CollectiveOptimizationFor},
    optimization based on trigonometric model functions \cite{Koczor2020QuantumAnalyticDescent},
    neural network-based optimization methods \cite{Rivera-Dean21AvoidingLocalMinima},
    brick-layer structures of generic unitaries \cite{Slattery22UnitaryBlockOptimization},
    and operator pool-based methods \cite{Grimsley19AnAdaptiveVariational,Boyd22TrainingVariationalQuantum}.
ADAPT-VQEs \cite{Grimsley19AnAdaptiveVariational} iteratively grow the \ac{VQA}'s \ac{PQC} by adding operators from a pool that have led to the largest derivative in the previous step.
This strategy allows one to avoid barren plateaus and even ``burrow'' out of some traps \cite{Grimsley22ADAPT-VQE}.
CoVar \cite{Boyd22TrainingVariationalQuantum} is based on similar ideas complemented with estimating several properties of the variational state in parallel using classical shadows \cite{Huang2020Predicting}.
The optimization strategies are of a heuristic nature, and analytic results are scarce. Finally, it has been numerically observed \cite{ToussiKiani20LearningUnitariesBy,Wiersema20ExploringEntanglementAnd} and analytically shown \cite{Larocca21TheoryOfOverparametrization}
that \acs{VQA}-type ansätze become almost free from traps when the ansatz is overparameterized.
{Our work implies that these practical approaches cannot work for all instances and, therefore, provides a justification to resort to such heuristics.}

\subsection{Techniques} \label{scn:techniques}
We focus on techniques for showing \QCMA-hardness of approximation, as containment in \QCMA\ is straightforward\footnote{The prover sends angles $\theta_j$, and the verifier simulates each $e^{i\theta_jH_j}$ via known Hamiltonian simulation algorithms~\cite{LC17}.} for both \VQA\ and \QAOA.

To begin, recall that in a \QCMA\ proof system (\Cref{def:QCMA}), given a YES input, there exists a poly-length \emph{classical} proof $y$ causing a quantum poly-size circuit $V$ to accept, and for a NO input, all poly-length proofs $y$ cause $V$ to reject. Our goal is to embed such proof systems into instances of \Cref{def:VQA} and \Cref{def:QAOA}, while maintaining a large promise gap ratio $m'/m$. To do so, we face three main challenges: (1) Where will hardness of approximation come from? Typically, one requires a PCP theorem~\cite{AS98,ALMSS98} for such results, which remains a notorious open question for both \QCMA\ and \QMA\footnote{Quantum Merlin-Arthur (\QMA) is \QCMA\ but with a quantum proof.}~\cite{AAV13}. (2) \Cref{def:VQA} places no restrictions on which Hamiltonians are applied, in which order, and with which rotation angles. How can one enforce computational structure given such flexibility? In addition, \QAOA\ presents a third challenge: (3) How to overcome the previous two challenges when we are only permitted two Hamiltonians, $H_b$ and $H_c$, the latter of which must also act as the observable?

To address the first challenge, we appeal to the hardness of approximation work of Umans~\cite{U99}. The latter showed how to use a graph-theoretical construct, known as a \emph{disperser}, to obtain strong hardness of approximation results for $\Sigma_2^p$ (the second level of the Polynomial-Time Hierarchy). Hiding at the end of that paper is Theorem 9, which showed that the techniques therein also apply to yield hardness of approximation within factor $N^{1/5-\epsilon}$ for a rather artificial NP-complete problem. Gharibian and Kempe~\cite{GK12} then showed that~\cite{U99} can be extended to obtain hardness of approximation results for a quantum analogue of $\Sigma_2^p$, and also obtained \QCMA-hardness of approximation within $N^{1-\epsilon}$ for an even more artificial problem, Quantum Monotone Minimum Satisfying Assignment (QMSA, \Cref{def:qmmww}). Roughly, QMSA asks --- given a quantum circuit $V$ accepting a monotone set (\Cref{def:monotone}) of strings, what is the smallest Hamming weight string accepted by $V$?
Here, our approach will be to construct many-one reductions from QMSA to \VQA\ and \QAOA, where we remark that maintaining the $N^{1-\epsilon}$ hardness ratio (i.e.\ making the reduction approximation-ratio-preserving) will require special attention.\\

\noindent\emph{1. The reduction for \VQA.} To reduce a given QMSA circuit $V=V_L\cdots V_1$ to a VQA instance $(\set{H_i},M,m,m')$, we utilize a ``hybrid Cook-Levin $+$ Kitaev'' circuit-to-Hamiltonian construction, coupled with a \emph{pair} of clocks (whereas Kitaev~\cite{KitSheVya02} requires only one clock). Here, a \emph{non-hybrid} (i.e. standard) circuit-to-Hamiltonian construction is a quantum analogue of the Cook-Levin theorem, i.e. a map from quantum circuits $V$ to local Hamiltonians $H_V$, so that there exists a proof $\ket{\psi}$ accepted by $V$ if and only if $H_V$ has a low-energy\footnote{By ``energy'' of a state $\ket{\psi}$ against Hamiltonian $H$, one means the expectation $\bra{\psi}H\ket{\psi}$, whose minimum possible value is precisely $\lambda_{\min}(H)$, i.e. the smallest eigenvalue of $H$.} ``history state'', $\ket{\psi_{\textup{hist}}}$.
A history state, in turn, is a quantum analogue of a Cook-Levin tableau, except that each time step of the computation is encoded in superposition via a clock construction of Feynman~\cite{FeynmanComputers}.
In contrast, our construction is ``hybrid'' in that it uses a clock register like Kitaev, but does not produce a history state in superposition over all time steps, like Cook-Levin.
A bit more formally, the Hamiltonians $\set{H_i}$ of our VQA instance act on four registers, $ABCD$, denoting proof ($A$), workspace ($B$), clock 1 ($C$), and clock 2 ($D$). To an honest prover, these Hamiltonians $\set{H_i}$ may be viewed as being partitioned into two sets: Hamiltonians for ``setting proof bits'', denoted $P$, and Hamiltonians for simulating gates from $V$, denoted $Q$. An example of a Hamiltonian in $P$ is
   \begin{align}
                P_j&\coloneqq X_{A_j}\otimes \ketbra{1}{1}_{C_{j}}\otimes \ketbra{1}{1}_{D_{\absD}}\label{eqn:flip}
   \end{align}
which says: If clock 1 (register $C$) is at time $j$ \emph{and} clock 2 (register $D$) is at time $\abs{D}$ (more on clock 2 shortly), then flip the $j$th qubit of register $A$ via a Pauli $X$ gate. An example of a Hamiltonian in $Q$ is
                \begin{align}
                    Q_j &\coloneqq (V_j)_{AB}\otimes\ketbra{01}{10}_{C_{\absA+j,\absA+j+1}}+ (V_j^\dagger)_{AB}\otimes\ketbra{10}{01}_{C_{\absA+j,\absA+j+1}},
                \end{align}
which allows the prover to apply gate $V_j$ of $V$ to registers $AB$, while updating clock 1 from time $\absA+j$ to $\absA+j+1$. In this first (insufficient) attempt at a reduction, the honest prover for \VQA\ acts as follows: First, apply a subset of the $P$ Hamiltonians to prepare the desired input $y$ to the QMSA verifier $V$ in register $A$, and then evolve Hamiltonians $Q_1$ through $Q_L$ to simulate gates $V_1$ through $V_L$ on registers $A$ and $B$. The observable $M$ is then defined to measure the designated output qubit of $B$ in the standard basis, conditioned on $C$ being at time $T$.

The crux of this (honest prover) setup is that if we start with a YES (respectively, NO) instance of QMSA, then the Hamming weight of the optimal $y$ is at most $g$ (respectively, at least $g'$), for $g'/g\geq \NQMSA^{1-\epsilon}$ and $\NQMSA$ the encoding size of the QMSA instance. This, in turn, means that the VQA prover applies at most $g$ Hamiltonians from $P$ (YES case), or at least $g'$ Hamiltonians from $P$ (NO case). The problem is that the prover must \emph{also} apply Hamiltonians $Q_1$ through $Q_L$ in order to simulate the verifier, $V$, and so we have hardness ratio $m'/m=(g'+L)/(g+L)\rightarrow 1$ if $L\in\omega(g)$, as opposed to $N^{1-\epsilon}$!

To overcome this, we make flipping each bit of $P$ ``more costly'' by utilizing a \emph{2D clock setup}.
This, in turn, will ensure the hardness ratio $(g'+L)/(g+L)$ becomes (roughly)
\begin{eqnarray}
    \frac{g'\absD+L}{g\absD+L}\approx \frac{g'}{g} \text{ for }\absD\in\omega(L),
\end{eqnarray}
as desired.
Specifically, to flip bit $A_j$ for any $j$, we force the prover to first sequentially increment the \emph{second} clock, $D$, from $1$ to $\absD$. By \Cref{eqn:flip}, $P_j$ can now flip the value of $A_j$ --- but it cannot increment time in $C$ (i.e.\ we remain in time step $j$ on clock 1). This next forces the prover to decrement $D$ from $\absD$ back to $1$, at which point a separate Hamiltonian (not displayed here) can increment clock $C$ from $j$ to $j+1$. The entire process then repeats itself to flip bit $A_{j+1}$. What is crucial for our desired approximation ratio is that we only have a single copy of register $D$, i.e.\ we re-use it to flip each bit $A_j$, thus effectively making $CD$ act as a 2D clock. This ensures the added overhead to the encoding size of the VQA instance scales as $\absD$, not $\absA\absD$, which is what one would obtain if $CD$ encoded a 1D clock (i.e. if each $A_j$ had a \emph{separate} copy of $D$).

Finally, to show soundness against provers deviating from the honest strategy above, we first establish that any sequence of evolutions from $\set{H_i}$ keeps us in a desired logical computation space, i.e.\ the span of vectors of form
            \begin{align}\label{eqn:logical}
                S&\coloneqq \set*{V_{s-\absA}\cdots V_{1}\ket{y}_A\ket{0\cdots 0}_B\ket{\wts}_C\ket{\wtt}_D \given y\in\set{0,1}^{\absA}, s\in\set{1,\ldots, \absC}, t\in\set{1,\ldots, \absD}},
            \end{align}
for $\ket{y}_A$ the ``proof string'' prepared via $P$-gates and $\wts$ and $\wtt$ the unary representations of time steps $s$ and $t$ in clocks 1 and 2, respectively.
We then show that applying too few Hamiltonian evolutions from $\set{H_i}$ results in a state with either no support on large Hamming weight strings $y$ (meaning the verifier $V$ must reject in the NO case), or no support on states with a fully executed verification circuit $V=V_L\cdots V_1$ (in which case we design $V$ to reject).\\

\noindent\emph{2. The reduction for \QAOA.} At a high level, our goal is to mimic the reduction to \VQA\ above. However, the fact that we have only {two} Hamiltonians at our disposal, $H_b$ (driving Hamiltonian) and $H_c$ (cost Hamiltonian), and no separate observable $M$, complicates matters. Very roughly, our aim is to \emph{alternate} even and odd steps of the honest prover's actions from \VQA, so that $H_b$ simulates the even steps, and $H_c$ the odd ones. To achieve this requires several steps:
\begin{enumerate}
    \item First, we modify the \VQA\ setup so that all the odd (respectively, even) local terms $H_i$ pairwise commute. This ensures that the actions of $\exp(i\theta H_b)$ and $\exp(i\theta H_c)$ can be analyzed, since $H_b$ and $H_c$ will consist of sums of (now commuting) $H_i$ terms.

    \item In \VQA, all Hamiltonians satisfied $H_i^2=I$, which intuitively means an honest prover could use $H_i$ to either act trivially ($\theta_i=0$) or perform some desired action ($\theta_i=\pi$). For \QAOA, we instead require a trick inspired by~\cite{Bittel21TrainingVariationalQuantum} --- we introduce certain local terms $G_j$ (\Cref{eqn:Ge}) with $3$-cyclic behavior. In words, the honest prover can induce \emph{three} logical actions from such $G_j$, obtained via angles $\theta_j\in\set{0,\pi/3,2\pi/3}$, respectively.

    \item We next add additional constraints to $H_b$ to ensure its unique ground state encodes the correct start state (see \Cref{eqn:psi} of \Cref{def:QAOA}). This is in contrast to \VQA, where the initial state $\ket{0\cdots 0}$ is fixed and independent of the $H_i$.
    \item Finally, the observable $M$ is added as a local term to $H_c$, but scaled larger than all other terms in $H_c$. This ensures that for any state $\ket{\psi}$, $\abs{\bra{\psi}H_c-M\ket{\psi}}$ is ``small'', so that measuring cost Hamiltonian $H_c$ once the \ac{QAOA} circuit finishes executing is ``close'' to measuring $M$.
\end{enumerate}

As for soundness, the high-level approach is similar to \VQA, in that we analyze a logical space of computation steps, akin to \Cref{eqn:logical}, and track Hamming weights of prepared proofs in this space. The analysis, however, is more involved, as the construction itself is more intricate than for \VQA. For example, a new challenge for our \QAOA\ construction is that evolving by a Hamiltonian (specifically, $H_c$) does \emph{not} necessarily preserve the logical computation space. We thus need to prove that we may ``round'' each intermediate state in the analysis back to the logical computation space, in which we can then track the Hamming weight of the proof $y$ (\Cref{lem:gate_decay}).

\subsection{Open questions}
We have shown that the optimal depth of a VQA or \ac{QAOA} ansatz is hard to approximate, even up to large multiplicative factors. A natural question is whether similar \emph{\NP-hardness} of approximation results for depth can be shown when (e.g.) the cost Hamiltonian in \ac{QAOA} is classical, such as in \cite{FarGolGut14}?
Since we aimed here to capture the strongest possible hardness result, i.e.\ for \QCMA, our Hamiltonians were necessarily not classical/diagonal.
Second, although our results are theoretical worst-case results, VQAs are of immense practical interest in the NISQ community.
Can one design good heuristics for optimal depth approximation which often work well in practice?
Third, can one approximate the optimal depth for QAOA on \emph{random} instances of a computational problem?
Here, for example, recent progress has been made by Basso, Gamarnik, Mei and Zhou~\cite{BGMZ22}, Boulebnane and Montanaro \cite{BA22}, and Anshu and Metger \cite{AM22}, which give analytical bounds on the success probability of QAOA at various levels and on random instances of various constraint satisfaction problems, for instance size $n$ going to infinity.
The bounds of \cite{AM22}, for example, show that even \emph{superconstant} depth (i.e. scaling as $o(\log\log n)$) is insufficient for QAOA to succeed with non-negligible probability  for a random spin model.
On a positive note, we remark that \cite{BA22} give numerical evidence (based on their underlying analytical bounds) that at around level 14, QAOA begins to surpass existing classical SAT solvers for the case of random 8-SAT.
Fourth, we have given the first natural QCMA-hard to approximate problems.
What other QCMA-complete problems can be shown hard to approximate?
A natural candidate here is the \emph{Ground State Connectivity} problem~\cite{GS14,GMV17,WBG20}, whose hardness of approximation we leave as an open question.
Finally, along these lines, can a PCP theorem for QCMA be shown as a first stepping stone towards a PCP theorem for QMA?

\subsection{Organization}
This paper is organized as follows. We begin with basic definitions and notation in \Cref{scn:def}. In \Cref{scn:VQA}, we show \Cref{thm:QCMA}. \Cref{scn:QAOA} shows \Cref{thm:QCMA_QAOA}.

\section{Basic definitions and notation}\label{scn:def}

We begin with notation, and subsequently define QCMA.

\subsection{Notation}
Throughout, the relation $\coloneqq $ denotes a definition, and $[n] \coloneqq \set{1,2,\dots,n}$.
We use $\length{x}$ to specify the length of a vector or string or the cardinality of set $x$.
The term $I_A$ denotes the identity operator/matrix on qubits with indices in register $A$.
By $\snorm{H}$ we denote the spectral norm of an operator $H$ acting on $\complex^d$, i.e.\ $\max_{\ket{\psi}\in\complex^d}\frac{\enorm{H\ket{\psi}}}{\enorm{\ket{\psi}}}$, for $\enorm{\,\cdot\,}$ the standard Euclidean norm.
The trace norm of an operator is denoted by $\|\cdot\|_{\mathrm{tr}}$. $e_i$ refers to a computational basis state.

\subsection{Complexity classes}

\begin{definition}[Quantum-classical Merlin-Arthur (\QCMA)]\label{def:QCMA}
Let $\Pi= (\Pi_{\mathrm{yes}}, \Pi_{\mathrm{no}})$ be a promise problem.
Then $\Pi\in \QCMA$ if and only if
there is a polynomial $p$ such that for any $x\in \Pi$ there exists a quantum circuit $V_x$ of size $p(\length{x})$ with one designated output qubit satisfying:
\begin{compactenum}[(i)]
\item 
If $x\in \Pi_{\mathrm{yes}}$ there exists a string $y\in \set{0,1}^{p(\length{x})}$ such that $\Pr[V_x \text{ accepts } y]\geq 2/3$
and
\item 
if $x\in \Pi_{\mathrm{no}}$ and all strings $y\in \set{0,1}^{p(\length{x})}$ it holds that $\Pr[V_x \text{ accepts } y]\leq 1/3$.
\end{compactenum}
\end{definition}

\noindent Often, it is helpful to separate the qubits into an a \emph{proof register} $A$, which contains the classical proof $\ket y$, and an \emph{ancilla/work register} $B$, which is initialized in the $\ket 0$ state.
Then the acceptance probability can be expressed as
\begin{equation}
    \Pr[V_x \text{ accepts } (x,y)]
    =
    \sandwichb{y;0}{V^{(n)\dagger}_x M^{(B_1)} V^{(n)}_x}{y;0}\, ,
\end{equation}
where the measurement is given by an operator $M^{(B_1)}$ acting on the first qubit of the work register $B$.

\QCMA\ was first defined in \cite{Aharonov02QuantumNP}, and satisfies $\NP\subseteq\QCMA\subseteq\QMA$. QCMA-complete problems include Identity Check on Basis States (i.e.\ ``does a quantum circuit act almost as the identity on all computational basis states?'')~\cite{WJB03} and Ground State Connectivity (GSCON) (i.e.\ is the ground space of a local Hamiltonian ``connected''?)~\cite{GS14}. The latter remains hard (specifically, $\QCMA_{\textup{EXP}}$-hard) in the $1$D translation-invariant setting~\cite{WBG20}.

\section{\texorpdfstring{\QCMA}{QCMA}-hardness of approximation for \texorpdfstring{\acp{VQA}}{VQAs}}\label{scn:VQA}

In this section, we show \Cref{thm:QCMA}. We begin in \Cref{sscn:vqadefs} with relevant definitions and lemmas. \Cref{sscn:VQA} proves \Cref{thm:QCMA}.

\subsection{Definitions and required facts}\label{sscn:vqadefs}

For convenience, we first restate \Cref{def:VQA}.
\defvqa*

We next require definitions and a theorem from \cite{GK12}.

\begin{definition}[Monotone set]\label{def:monotone}
    A set $S\subseteq\set{0,1}^n$ is called \emph{monotone} if for any $x\in S$, any string obtained from $x$ by flipping one or more zeroes in $x$ to one is also in $S$.
\end{definition}

\begin{definition}[Quantum circuit accepting monotone set]
    Let $V$ be a quantum circuit consisting of $1$- and $2$-qubit gates, which takes in an $n$-bit classical input register, $m$-qubit ancilla register initialized to all zeroes, and outputs a single qubit, $q$. For any input $x\in\set{0,1}^n$, we say $V$ \emph{accepts} (respectively, \emph{rejects}) $x$ if measuring $q$ in the standard basis yields $1$ (respectively, $0$) with probability at least $1-\epsilon_Q$ (If not specified, $\epsilon_Q=1/3$). We say $V$ accepts a \emph{monotone set} if the set $S\subseteq\set{0,1}^n$ of all strings accepted by $V$ is monotone (\Cref{def:monotone}).
\end{definition}
\begin{problem}[\QMMWW~(QMSA)]\label{def:qmmww}
    Given a quantum circuit $V$
    accepting
    a non-empty monotone set $S\subseteq\set{0,1}^n$, and integer thresholds $0\leq g\leq g'\leq n$, output:
    \begin{itemize}
        \item YES if there exists an $x\in\set{0,1}^n$ of Hamming weight at most $g$ accepted by $V$.
        \item NO if all $x\in\set{0,1}^n$ of Hamming weight at most $g'$ are rejected by $V$.
    \end{itemize}
\end{problem}

\begin{theorem}[Gharibian and Kempe \cite{GK12}]\label{thm:QMSAhard}
        QMSA is QCMA-complete, and moreover it is QCMA-hard to
        decide whether, given an instance of QMSA, the minimum Hamming weight string accepted by $V$ is at most $g$ or at least $g'$ for $g'/g\in O(N^{1-\epsilon})$ (where $g'\geq g$).
\end{theorem}
\noindent In words, QMSA is QCMA-hard to
approximate within $N^{1-\epsilon}$ for any constant $\epsilon>0$, where $N$ is the encoding size of the QMSA instance.

\subsection{QCMA-completeness}\label{sscn:VQA}

\mainvqa*

\noindent In words, it is QCMA-hard to decide whether, given an instance of \VQA, the variational circuit can prepare a ``good'' ansatz state with at most $m$ evolutions, or if all sequences of $m'$ evolutions fail to prepare a ``good'' ansatz state, for $m'/m\in O(N^{1-\epsilon})$ (where $m'\geq m$).

\begin{proof}
       Containment in QCMA is straightforward; the prover sends the angles $\theta_i$ and indices of Hamiltonians $H_i$ to evolve, which the verifier then completes using standard Hamiltonian simulation techniques~\cite{Llo96,LC17}. We now show QCMA-hardness of approximation. Let $\Pi'=(V',g,g')$ be an instance of QMSA, for $V'=V'_{L'}\cdots V'_1$ a sequence of $L'$ $2$-qubit gates taking in $\n'$ input bits and $\m'$ ancilla qubits.

       \paragraph{Preprocessing $V'$.} To ease our soundness analysis, we would like to make two assumptions about $V'$ without loss of generality; these can be simply ensured as follows. Suppose $V'$ takes in $\n'$ input qubits in register $A'$ and $\m'$ ancilla qubits in register $B'$. Apply each of the following modifications in the order listed.

\begin{ass}\label{ass:1}
            $V'$ only \emph{reads} register $A'$, but does not write to it. To achieve this, add $\n'$ ancilla qubits (initialized to $\ket{0}$) to $B'$, and prepend $V'$ with $\n'$ CNOT gates applied transversally to copy input $x$ from $A'$ to the added ancilla qubits in $B'$. Update any subsequent gate which acts on the original input $x$ to instead act on its copied version in $B'$.
\end{ass}
\begin{ass}\label{ass:2}
The output qubit of $V'$ is set to $\ket{0}$ until $V'_L$ is applied. To achieve this, add a single ancilla qubit to $B'$ initialized to $\ket{0}$, and treat this as the new designated output qubit. Append to the end of $V'$ a CNOT gate from its original output wire to the new output wire.
\end{ass}
        \noindent Call the new circuit with all modifications $V$. $V$ acts on $\n\coloneqq \n'$ input qubits, $\m\coloneqq \m'+\n'+1$ ancilla qubits, and consists of $L\coloneqq L'+\n'+1$ gates.

\paragraph{Proof organization.} The remainder of the proof is organized as follows. Section \ref{sscn:VQAinstance} constructs the \VQA\ instance. Section \ref{sscn:VQAlemmas} proves observations and lemmas required for the completeness and soundness analyses. Sections \ref{sscn:VQAcompleteness} and \ref{sscn:VQAsoundness} show completeness and soundness, respectively. Finally, Section \ref{sscn:VQAratio} analyzes the hardness ratio achieved by the reduction.

        \subsubsection{The \VQA\ instance} \label{sscn:VQAinstance}

\begin{figure}
	\centering
	\includegraphics[width=120mm]{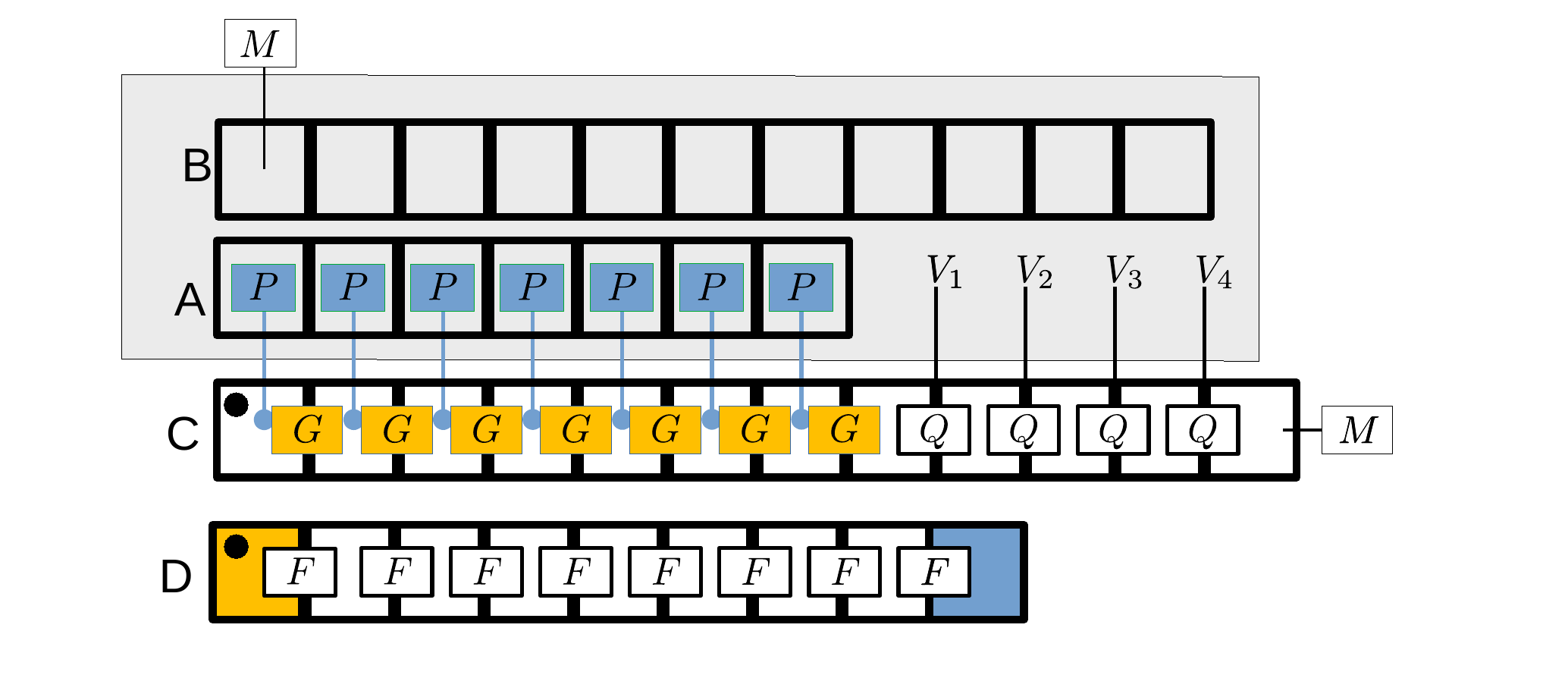}
	\caption{Sketch describing the VQA instance.  A colored square (say, blue) at index $j$ of a register means that register's $j$th qubit must be in $\ket{1}$ for any blue gates to act non-trivially. So, for example, the $G$ gates increment the first clock register $C$, but only if the $D$ register is in the state $\ket{1}_{D_1}$. For the initial state, $C_1$ and $D_1$ are in the $\ket{1}$ state, marked by a black dot. The gates $F$ increment the second clock register $D$. The $P$ gates are controlled operations on the $C$ register, which perform $X$ operations on the $A$ register, but only if $D$ is in the state $\ket{1}_{D_{\absD}	}$. The $Q$ gates increment the clock register $C$, while also applying the circuit $V_1,\dots,V_L$ on the $AB$ registers. The measurement operator $M$ acts on the $B_1$ and $C_{\absC}$ qubit.
    }
	\label{fig:sketch}
\end{figure}

        We now construct our instance $\Pi$ of \VQA\ as follows. $\Pi$ acts on a total of $n$ qubits, which we partition into $4$ registers: $A$ (proof), $B$ (workspace), $C$ (clock 1), and $D$ (clock 2). Register $A$ consists of $\n$ qubits, $B$ of $\m$ qubits, $C$ of $L+\n+1$ qubits, and $D$ of $\lceil L^{1+\delta}\rceil$ qubits for some fixed $0<\delta<1$ to be chosen later.\ Throughout, we use shorthand (e.g.) $\absA$ for the number of qubits in a register $A$.

        Our construction will ensure $C$ (respectively, $D$) always remains in the span of logical time steps, $\Tc\coloneqq \set{\ket{\wts}}_{s=1}^{\absC}$ (respectively, $\Td\coloneqq \set{\ket{\wtt}}_{t=1}^{\absD}$), defined as:
        \begin{align}
               \ket{\wts}\coloneqq \ket{0}^{\otimes s-1}\ket{1}\ket{0}^{\otimes \absC-s}\quad\text{for }1\leq s\leq \absC \label{eqn:tildes}\\
               \ket{\wtt}= \ket{0}^{\otimes t-1}\ket{1}\ket{0}^{\otimes \absD-t}\quad\text{for }1\leq t\leq \absD. \label{eqn:tildet}
        \end{align}

        For example for $C$, $\ket{\wto}=\ket{1}\ket{0}^{\otimes {\absC-1}}$, $\ket{\wt{2}}=\ket{0}\ket{1}\ket{0}^{\otimes {\absC-2}}$, $\ket{\wt{3}}=\ket{00}\ket{1}\ket{0}^{\otimes {\absC-3}}$, and so forth. Note this differs from the usual Kitaev unary clock construction, which encodes time $t$ via $\ket{1}^{\otimes t}\ket{0}^{\otimes N-t}$~\cite{KitSheVya02}. This allows us to reduce the locality of our Hamiltonian.

        Throughout, we use (e.g.) $C_j$ to refer to qubit $j$ and $C_{i,j}$  and qubits $i$ and $j$ of register $C$. All qubits not explicitly mentioned are assumed to be acted on by the identity.
        Define four families of Hamiltonians as follows:

        \begin{itemize}
            \item ($F$) For propagation of the second clock, $D$, define $2$-local Hamiltonians as
            \begin{align}
                F_j&\coloneqq \ketbra{01}{10}_{D_{j,j+1}}+\ketbra{10}{01}_{D_{j,j+1}} \text{ for all }j\in \set{1,\ldots, \absD-1}.\label{eqn:F}
            \end{align}
            \item ($G$) For propagation of the first clock, $C$, define $3$-local Hamiltonians as
            \begin{align}
                G_j&\coloneqq \left(\ketbra{01}{10}_{C_{j,j+1}}+\ketbra{10}{01}_{C_{j,j+1}}\right)\otimes \ketbra{1}{1}_{D_1} \text{ for all }j\in \set{1,\ldots, \absA}.\label{eqn:G}
            \end{align}
            \item ($P$) For each qubit $j\in\set{1,\ldots, \absA}$ of $A$, define $3$-local Hamiltonian as
            \begin{align}
                P_j&\coloneqq  X_{A_j}\otimes \ketbra{1}{1}_{C_{j}}\otimes \ketbra{1}{1}_{D_{\absD}}.
            \end{align}
            \item ($Q$) For each gate $V_k$ for $k\in\set{1,\ldots, L}$, let $R_k$ denote the two qubits of $AB$ which $V_k$ acts on.
            Define $4$-local Hamiltonians as
                \begin{align}
                    Q_k &\coloneqq  (V_k)_{R_k}\otimes\ketbra{01}{10}_{C_{\absA+k,\absA+k+1}}+ (V_k^\dagger)_{R_k}\otimes\ketbra{10}{01}_{C_{\absA+k,\absA+k+1}}.\label{eqn:Qt}
                \end{align}
        \end{itemize}
        Denote the union of these four sets of Hamiltonians as $\SFGPQ\coloneqq F\cup G\cup P\cup Q$.
        Set a $2$-local observable
        \begin{equation}\label{eq:def:M}
            M\coloneqq I-\ketbra{1}{1}_{B_1}\otimes \ketbra{1}{1}_{C_{\absC}}
        \end{equation}
        where we assume without loss of generality that $V$ outputs its answer on qubit $B_1$.
        Set $m=g\cdot (2\absD-1)+\absA+L$, $m'=g'\cdot (2\absD-1)+\absA+L$.
        To aid the reader in the remainder of the proof, all definitions above are summarized in Figure~\ref{fig:defs1}.
        \begin{figure}[t]\label{fig:defs1}
        \centering
            \begin{tabular}{|c|l|l|}
            \hline
              Term & Description & Properties \\
            \hline
              $V'$ & Input QMSA instance's verification circuit & $V'=V'_{L'}\cdots V'_1$\\
              \hline
              $L'$ & Number of $1$- and $2$-qubit gates in $V'$&\\
              \hline
              $\n'$ & Number of proof qubits taken in by $V'$ &\\
              \hline
              $\m'$ & Number of ancilla qubits taken in by $V'$ &\\
              \hline
              $g$,$g'$ & YES/NO thresholds for QMSA instance, resp. &\\
              \hline
              $V$ & QMSA verifier obtained from $V'$ via Assump. \ref{ass:1} and \ref{ass:2}&$V=V_{L}\cdots V_1$\\
              \hline
              $L$ & Number of $1$- and $2$-qubit gates in $V$& $L=L'+\n'+1$\\
              \hline
              $\n$ & Number of proof qubits taken in by $V$ & $\n=\n'$\\
              \hline
              $\m$ & Number of ancilla qubits taken in by $V$ & $\m=\m'+\n'+1$\\
              \hline
              $A$ & Proof register & $\abs{A}=\n$ \\
              \hline
              $B$ & Workspace register & $\abs{B}=\m$ \\
              \hline
              $C$ & Clock 1 register &  $\abs{C}=L+\n+1$\\
              \hline
              $D$ & Clock 2 register &  $\abs{D}=\lceil L^{1+\delta}\rceil$ for $\delta$ chosen in\\
              && (\ref{eqn:goal2}) to satisfy (\ref{eqn:goal}) \\
              \hline
              $F$ & Propagation terms for clock 2 & Act on register $D$, \\
               &  & $\abs{F}=\abs{D}-1$ \\
               \hline
              $G$ & Propagation terms for clock 1 & Act on registers $C,D$, \\
               &  & $\abs{g}=\abs{A}$ \\
               \hline
              $P$ & Hamiltonian terms for setting proof bits & Act on registers $A,C,D$, \\
                             &  & $\abs{P}=\abs{A}$ \\
               \hline
              $Q$ & Hamiltonian terms for simulating verifier gates, $V_k$ & Act on registers $A,B,C$, \\
               &  & $\abs{Q}=L$ \\
               \hline
               $M$ & Observable for \VQA\ instance  & $M\coloneqq I-\ketbra{1}{1}_{B_1}\otimes \ketbra{1}{1}_{C_{\absC}}$\\
                \hline
                 $m$,$m'$ & YES/NO thresholds for \VQA\ instance, resp. &
                $m=g\cdot (2\absD-1)+\absA+L$,\\
                &&$m'=g'\cdot (2\absD-1)+\absA+L$.\\
                \hline
            \end{tabular}
            \caption{Terms used in the proof of \Cref{thm:QCMA}.}
        \end{figure}

        It remains to choose our initial state. Strictly speaking, \Cref{def:VQA} mandates initial state $\ket{0\cdots 0}_{ABCD}$. However, to keep notation simple, it will be convenient to instead choose
        \begin{align}\label{eq:def:phi}
            \ket{\phi}\coloneqq \ket{0\cdots 0}_{AB}\ket{10^{\absC-1}}_C\ket{10^{\absD-1}}_D=\ket{0\cdots 0}_{AB}\ket{\wto}_C\ket{\wto}_D,
        \end{align}
        i.e.\ with the two clock registers $C$ and $D$ initialized to their starting clock state, $\ket{\wto}$. This is without loss of generality --- we may, in fact, start with \emph{any} standard basis state as our initial state without requiring major structural changes to our construction, as the following observation states.
        \begin{obs}\label{obs:startstate}
            Fix any standard basis state $\ket{x}_{ABCD}=\overline{X}\ket{0\cdots 0}_{ABCD}$, for
            \begin{align}
                \overline{X}\coloneqq X_1^{x_1}\otimes\cdots\otimes X_{N}^{x_{N}}
            \end{align}
            with $N\coloneqq \absA+\absB+\absC+\absD$.
            Consider the updated set $\SFGPQ'\coloneqq \set{\overline{X}H\overline{X}\mid H\in\SFGPQ}$, where for simplicity we match $H\in\SFGPQ$ with $H'\coloneqq \overline{X}H\overline{X}\in\SFGPQ'$. Then, for any $m\in\nats$, and any sequence $(H_t)_{t=1}^m$ of Hamiltonians drawn from $\SFGPQ$,
            \begin{align}
                 e^{i\theta_mH_m}\cdots e^{i\theta_2H_2}e^{i\theta_1H_1}\ket{x}_{ABCD} = e^{i\theta_mH'_m}\cdots e^{i\theta_2H'_2}e^{i\theta_1H'_1}\ket{0\cdots 0}_{ABCD}.
            \end{align}
            Moreover, each $H$ and $H'$ are the same locality.
        \end{obs}
        \begin{proof}
            The first claim follows since $X^2=I$, by \Cref{obs:closedform}\footnote{Note that conjugation by $\overline{X}$ will alter the clock encoding in Equations (\ref{eqn:closedformF})-(\ref{eqn:closedformQ}), but this alternation is logically irrelevant since it is equivalent to a local change of basis applied simultaneously to all $H\in\SFGPQ$ and to $\ket{0\cdots 0}_{ABCD}$.}, and since $\ket{x}_{ABCD}=\overline{X}\ket{0\cdots 0}_{ABCD}$. The second claim also follows from $X^2=I$ and the fact that $\overline{X}$ is a tensor product of operators from set $\set{I,X}$.
        \end{proof}
        This concludes the construction.

        \subsubsection{Helpful observations and lemmas}\label{sscn:VQAlemmas}
        We next state and prove all observations and technical lemmas for the later correctness analysis of our construction.

        \begin{obs}\label{obs:closedform}
            For all $\theta\in \reals$, and all $F_j\in F$, $G_j\in G$, $P_j\in P$ and $Q_k\in Q$,
            \begin{align}
                e^{i\theta F_j} &= & &\cos(\theta)(\ketbra{01}{01}+\ketbra{10}{10})_{D_{j,j+1}} + i\sin(\theta)F_j+ (I-\ketbra{01}{01}-\ketbra{10}{10})_{D_{j,j+1}}\label{eqn:closedformF}\\
                e^{i\theta G_j} &= & &\cos(\theta)\left(\ketbra{01}{10}_{C_{j,j+1}}+\ketbra{10}{01}_{C_{j,j+1}}\right)\otimes \ketbra{1}{1}_{D_1}  + i\sin(\theta)G_j+\nonumber\\ &&&\left(I-\left(\ketbra{01}{10}_{C_{j,j+1}}+\ketbra{10}{01}_{C_{j,j+1}}\right)\otimes \ketbra{1}{1}_{D_1} \right)\label{eqn:closedformG}\\
                e^{i\theta P_j} &= & &\left(\cos(\theta) I+ i \sin(\theta) X\right)_{A_j}\otimes \ketbra{1}{1}_{C_{j}}\otimes\ketbra{1}{1}_{D_{\absD}} + (I-\ketbra{1}{1}_{C_{j}}\otimes\ketbra{1}{1}_{D_{\absD}})\label{eqn:closedformP}\\
                e^{i\theta Q_k} &= & &\cos(\theta) I_{AB}\otimes (\ketbra{01}{01}+\ketbra{10}{10})_{C_{\absA+k,\absA+k+1}}+i\sin(\theta)Q_{k}+\nonumber\\
                                & & &I_{AB}\otimes (I-\ketbra{01}{01}-\ketbra{10}{10})_{C_{\absA+k,\absA+k+1}}.\label{eqn:closedformQ}
            \end{align}
            For clarity, any register not explicitly listed in equations above is assumed to be acted on by identity.
        \end{obs}
        \begin{proof}
            Follows straightforwardly via Taylor series expansion since
            \begin{align}
                F_j^2&=(\ketbra{01}{01}+\ketbra{10}{10})_{D_{j,j+1}}
                    &\text{for all }j\in \set{1,\ldots, \absD-1},\\
                G_j^2&=(\ketbra{01}{01}+\ketbra{10}{10})_{C_{j,j+1}}\otimes \ketbra{1}{1}_{D_1}
                    &\text{for all }j\in \set{1,\ldots, \absA},\\
                P_j^2&=I_{A_j}\otimes \ketbra{1}{1}_{C_{j}}\otimes \ketbra{1}{1}_{D_{\absD}}
                    &\text{for all } j\in \set{1,\dots, \abs A},
                \\
                Q_k^2&=I_{AB}\otimes (\ketbra{01}{01}+\ketbra{10}{10})_{C_{\absA+k,\absA+k+1}}
                    & \text{for all } k\in \set{1,\dots, L}\, .
            \end{align}
        \end{proof}
        \begin{definition}[Support only on logical time steps]\label{def:supported}
            We say state $\ket{\psi}_{ABCD}$ is \emph{supported only on logical time steps} if it can be written
            \begin{equation}
                \ket{\psi}_{ABCD}=\sum_{s=1}^{\absC}\sum_{t=1}^{\absD}\alpha_{st}\ket{\eta_{st}}_{AB}\ket{\wts}_C\ket{\wtt}_D
            \end{equation}
            for unit vectors $\ket{\eta_{st}}$ and $\sum_{st}\abs{\alpha_{st}}^2=1$, and $\ket{\wts}\in \Tc$ and $\ket{\wtt}\in\Td$ defined as in \Cref{eqn:tildes} and \Cref{eqn:tildet}, respectively.
        \end{definition}
        \begin{obs}\label{obs:time}
            Recall that the initial state $\ket{\phi}= \ket{0\cdots 0}_{AB}\ket{\wto}_C\ket{\wto}_D$ is supported only on logical time steps. 
            Then, for any $m\in\nats$ and sequence of evolutions $\exp(i\theta_jH_j)$ for $\theta_j\in\reals$ and $H_j\in\SFGPQ$,
            \begin{equation}
                 e^{i\theta_mH_m}\cdots e^{i\theta_2H_2}e^{i\theta_1H_1}\ket{\phi}
            \end{equation}
             is supported only on logical time steps.
        \end{obs}
        \begin{proof}
            Consider any logical time step $\ket{\wts}_C\ket{\wtt}_D$. By Equations (\ref{eqn:closedformF})-(\ref{eqn:closedformQ}), the set of possible evolutions act as follows\footnote{Without loss of generality, we focus on the \emph{non-trivial} (i.e.\ non-identity) action of each evolution, as any trivial action immediately preserves logical time steps.}:
             \begin{itemize}
                \item $e^{i\theta F_j}$: can map from $\ket{\wt{j}}_D$ to $\ket{\wt{j+1}}_D$ or vice versa for $j\in\set{1,\ldots, \absD-1}$.
                \item $e^{i\theta G_j}$: if $\ket{\wtt}_D=\ket{\wto}_D$, can map from $\ket{\wt{j}}_C$ to $\ket{\wt{j+1}}_C$ or vice versa for $j\in\set{1,\ldots, \absA}$.
                \item $e^{i\theta P_j}$: acts invariantly (i.e.\ as identity) on $C$ and $D$ for all $j\in\set{1,\ldots, \absA}$.
                \item $e^{i\theta Q_j}$: can map from $\ket{\wt{\absA+j}}_C$ to $\ket{\wt{\absA+j+1}}_C$ (while applying $V_j$ to $AB$) or vice versa (for $V_j^\dagger$) for $j\in\set{1,\ldots, L}$.\qedhere
             \end{itemize}
        \end{proof}

The following lemma tells us that any sequence of Hamiltonian evolutions $\exp(i\theta_u H_u)$ on initial state $\ket{\phi}$ remains in a certain logical computation space.

        \begin{lemma}\label{l:span}
            Define
            \begin{align}
                S&\coloneqq \set*{V_{s-\absA}\cdots V_{1}\ket{y}_A\ket{0\cdots 0}_B\ket{\wts}_C\ket{\wtt}_D \given y\in\set{0,1}^{\absA}, s\in\set{1,\ldots, \absC}, t\in\set{1,\ldots, \absD}},\label{eqn:7}
            \end{align}
            where we adopt the convention that the $V$ gates are present only when $s> \absA$.
            Then, for any $m\in\nats$,
            \begin{align}
                \Pi_{u=1}^me^{i\theta_u H_u}\ket{\phi}\in \Span(S)
            \end{align}
            for any angles $\theta_u\in \reals$ and sequence of Hamiltonians $H_u\in \SFGPQ$.
        \end{lemma}
        \begin{proof}
    For convenience, define
    \begin{equation}\label{eqn:etays}
        \etayst\coloneqq V_{s-\absA}\cdots V_1\ket{y}_A\ket{0}_B\ket{\wts}_C\ket{\wtt}_D
    \end{equation}
    for $s\in\set{1,\ldots, \absC}$, $t\in\set{1,\ldots, \absD}$, and $y\in\set{0,1}^{\absA}$. Observe first that $\ket{\phi}=\ket{0\cdots 0}_{AB}\ket{\wto}_C\ket{\wto}_D=\etaoo\in S$. Thus, it suffices to prove $\Span(S)$ is closed under application of $e^{i\theta H}$ for any $\theta\in \reals$ and $H\in \SFGPQ$.

\paragraph{Case 1: $H=H_j\in F$ for $j\in\set{1,\ldots, \absD-1}$.} Equations (\ref{eqn:F}) and~(\ref{eqn:closedformF}) immediately yield $e^{i\theta H}\etayst=\etayst$ for $t\not\in\set{j,j+1}$. Consider thus  $t\in\set{j,j+1}$. Restricted to this space, $e^{i\theta H}$ acts logically as
    \begin{equation}
        e^{i\theta H} =\cos(\theta) I_{ABCD}+ i\sin(\theta)I_{ABC}\otimes\left(\ketbra{\wt{j+1}}{\wt{j}}+\ketbra{\wt{j}}{\wt{j+1}}\right)_D.
    \end{equation}
        Thus, $e^{i\theta H}$ maps
    \begin{eqnarray}
        \ket{\eta_{y,s,j}}&\mapsto&\cos(\theta)\ket{\eta_{y,s,j}}+i\sin(\theta)\ket{\eta_{y,s,j+1}},\\
        \ket{\eta_{y,s,j+1}}&\mapsto&i\sin(\theta)\ket{\eta_{y,s,j}}+\cos(\theta)\ket{\eta_{y,s,j+1}}.
    \end{eqnarray}

\paragraph{Case 2: $H=H_j\in G$ for $j\in\set{1,\ldots, \absA}$.} Equations (\ref{eqn:G}) and~(\ref{eqn:closedformG}) immediately yield $e^{i\theta H}\etayst=\etayst$ unless $s\in\set{j,j+1}$ and $t=1$. Consider thus  $s\in\set{j,j+1}$ and $t=1$. Restricted to this space, $e^{i\theta H}$ acts logically as
    \begin{equation}
        e^{i\theta H} =\cos(\theta) I_{ABCD}+ i\sin(\theta)I_{AB}\otimes\left(\ketbra{\wt{j+1}}{\wt{j}}+\ketbra{\wt{j}}{\wt{j+1}}\right)_C\otimes I_D.
    \end{equation}
        Thus, $e^{i\theta H}$ maps
    \begin{eqnarray}
        \ket{\eta_{y,j,1}}&\mapsto&\cos(\theta)\ket{\eta_{y,j,1}}+i\sin(\theta)\ket{\eta_{y,j+1,1}},\\
        \ket{\eta_{y,j+1,1}}&\mapsto&i\sin(\theta)\ket{\eta_{y,j,1}}+\cos(\theta)\ket{\eta_{y,j+1,1}}.
    \end{eqnarray}
    \paragraph{Case 3: $H=H_j\in P$ for $j\in\set{1,\ldots, \absA}$.} By Equation (\ref{eqn:closedformP}), $e^{i\theta H}\etayst=\etayst$ unless $s=j$ and $t=\absD$. Consider thus  $s=j$ and $t=\absD$. Restricted to this space, $e^{i\theta H}$ maps
    \begin{equation}\label{eqn:HWup}
        \ket{\eta_{y,j,\absD}}\mapsto \cos(\theta)\ket{\eta_{y,j,\absD}}+i\sin(\theta)\ket{\eta_{y',j,\absD}}
    \end{equation}
    for $y'$ defined as $y$ with its $j$th bit flipped. Since $y$ in \Cref{eqn:7} is not fixed, we conclude $e^{i\theta H}\etayst\in\Span(S)$, as claimed.

    \paragraph{Case 4: $H=H_j\in Q$.} Equations (\ref{eqn:Qt}) and~(\ref{eqn:closedformQ}) immediately yield $e^{i\theta H}\etayst=\etayst$ for $s\not\in\set{\absA+j,\absA+j+1}$. Consider thus  $s\in\set{\absA+j,\absA+j+1}$. Restricted to this space, $\exp(i\theta H)$ acts logically as
    \begin{equation}
        e^{i\theta H} =\cos(\theta) I_{ABCD}+i\sin(\theta)Q_j.
    \end{equation}
    Thus, $e^{i\theta H}$ maps
    \begin{align}
        \ket{\eta_{y,\absA+j,t}}&\mapsto\cos(\theta)\ket{\eta_{y,\absA+j,t}}+i\sin(\theta)V_j\ket{\eta_{y,\absA+j,t}}\nonumber\\
        &=\cos(\theta)\ket{\eta_{y,\absA+j,t}}+i\sin(\theta)\ket{\eta_{y,\absA+j+1,t}},\label{eqn:10}\\
        \ket{\eta_{y,\absA+j+1,t}}&\mapsto i\sin(\theta)V_j^\dagger\ket{\eta_{y,\absA+j+1,t}}+\cos(\theta)\ket{\eta_{y,\absA+j+1,t}}\nonumber\\
        &=i\sin(\theta)\ket{\eta_{y,\absA+j,t}}+\cos(\theta)\ket{\eta_{y,\absA+j+1,t}}\label{eqn:11},
    \end{align}
    where we have used \Cref{ass:1} that gates $V_j$ never write to the proof register $A$ (and thus $y$ remains unchanged under application of $V_k$). This yields the claim.
\end{proof}

{Next, we relate the circuit depth of a state generated by our VQA to the Hamming weight of the proof string $y$.}
\begin{lemma}\label{lem:HW}
    Let $(H_u)_{u=1}^m$ be a sequence of Hamiltonians drawn from $\SFGPQ$ which maps the initial state \eqref{eq:def:phi} to
    \begin{align}
        \ket{\phi_m}\coloneqq \Pi_{u=1}^me^{i\theta_u H_u}\ket{\phi}.
    \end{align}
    Suppose $\ket{\phi_m}$ has non-zero overlap with some $\etayst$ with $y$ of Hamming weight at least $w$ and $s=\abs{A}+1$. Then, $m\geq w(2\absD-1)+\absA$ with at least $ w(2\absD-1)+\absA$ of the $H_u$ drawn from $F\cup G\cup P$.

\end{lemma}

\begin{proof}
    By \Cref{l:span}, $\ket{\phi_m}\in S$.
    Repeating the following argument for any bit of $y$ set to $1$ will yield a lower bound on the number of gates of $2w\absD$, which is almost what we want.

    Consider any index $j\in\set{1,\ldots,\absA}$ such that $y_j=1$ (equivalently, in state $\etayst$ the qubit $A_j$ is set to $1$).
    Since the qubit $A_j$ of the
    initial state $\ket{\phi}$ is set to $\ket{0}$, there must be an evolution step $u\in\set{1,\ldots, m}$ at which $A_j$ is mapped from\footnote{For clarity, throughout this proof, by ``mapped from $\ket{k}$ to $\ket{k'}$'', we mean $\ket{k}$ is mapped to a state with non-zero overlap with $\ket{k'}$. This suffices for \Cref{lem:HW}, since it only cares about non-zero overlap with some $\etayst$.} $\ket{0}$ to $\ket{1}$.
    By \Cref{obs:closedform}, only the Hamiltonian $P_j\in P$ can induce this mapping, and $P_j$ requires $C$ and $D$ to be set to $\ket{\wt{j}}_C\ket{\wt{\absD}}_D$ in order to act non-trivially.
    Let us analyze each of these two requirements in order.

    First, to obtain $\ket{\wt{j}}$ in $C$, there are only two possibilities:
    \begin{itemize}
        \item (Case a) We are in the initial state $\ket{\phi}$ with no Hamiltonian evolutions applied yet and $j=1$ (recall $\ket{\phi}$ has $C$ and $D$ set to $\ket{\wto}_C\ket{\wto}_D$ by definition), or
        \item (Case b) we are at a later evolution step at which $C$ was updated to $\ket{\wt{j}}$ from either $j-1$ or $j+1$. Since $1\leq j\leq\absA$, by \Cref{obs:closedform}, only operators $G_{j-1}$ (for $j>1$) and $G_{j}$ can effect this map. Both of these operators require $D$ set to $\ket{\wto}$.
    \end{itemize}
    Thus, in both cases, $D$ is also set to $\ket{\wto}$. So, assume one of the two cases has just occurred to update $C$ to $\ket{\wt{j}}$.

    The second requirement for $P_j$ to act non-trivially was that $D$ be set to $\ket{\wt{\absD}}$. But since $D$ is
    currently
    set to $\ket{\wto}$ in the initial state, and since only operators in $F$ can change clock $D$ (by precisely one time step per operator), we must apply at least $\absD-1$ operators in $F$ to obtain a state with $\ket{\wt{j}}_C\ket{\wt{\absD}}_D$. (To see that $C$ must still be set to $\ket{\wt{j}}_C$ at this point, observe that all operators in $\SFGPQ\setminus F$ act invariantly unless $D$ equals $\ket{\wto}$ or $C$ is at least $\ket{\wt{\absA+1}}$.) Applying $P_j$ is now necessary to flip $A_j$ from $\ket{0}$ to $\ket{1}$. We have thus reached an intermediate state at which $A_j$ is $\ket{1}$ and $C$ and $D$ are $\ket{\wt{j}}_C\ket{\wt{\absD}}_D$.

    Finally, either all bits of $y$ are now set correctly and $C$ must be updated to $\ket{\wt{\absA+1}}$ (due to the condition
    $s=\absA+1$), or we wish to repeat the argument above for the next index $j'\neq j$ for which we wish to map $A_{j'}$ from $\ket{0}$ to $\ket{1}$. In both cases, $D$ must first be reset back to $\ket{\wto}$ (otherwise operators in $G$ act invariantly, and these are precisely the operators which can update $C$ to either $j'$ or to $\absA+1$ as needed). Running the argument above regarding $F$ in reverse, we obtain that at least a number $\absD-1$ of $F$-gates are needed to return $D$ back to $\ket{\wto}$, and at least one $G$-gate is needed to update $C$ from $j$ to $j'$ or to $\absA+1$.

    Summing all gate costs together, for each $A_j$ to be flipped from $\ket{0}$ to $\ket{1}$ and for $C$ to be updated to the next value of $j'$, we require at least $2\absD$ gates. Thus, if $\etayst$ has $y$ with Hamming weight at least $w$, at least $2w\absD$ gates from $F\cup G\cup P$ are required. This is almost what we want.

    The final $\absA-w$ gates required for the claim arise because one requires at least $\absA$ $G$-gates to map $C$ from its initial value of $1$ to $\absA+1$, and above we have only accounted for $w$ of these $G$-gates (i.e.\ corresponding to all $j$ with $y_j=1$).
\end{proof}

Finally, the next lemma ensures that any prover applying fewer than $L$ Hamiltonians from $Q$ cannot satisfy the YES case's requirements for \VQA.

\begin{lemma}\label{lem:L}
    For any $m\in\nats$, let $(H_u)_{u=1}^m$ be any sequence of Hamiltonians drawn from $\SFGPQ$ and containing strictly fewer than $L$ Hamiltonians from $Q$. Then, for observable $M=I-\ketbra{1}{1}_{B_1}\otimes \ketbra{1}{1}_{C_{\absC}}$, the state $\ket{\phi_m}\coloneqq \Pi_{u=1}^me^{i\theta_u H_u}\ket{0\cdots 0}_{ABC}$ satisfies
\begin{equation}\label{eqn:zero}
    \bra{\phi_m}M\ket{\phi_m}=1.
\end{equation}
\end{lemma}
\begin{proof}
    By \Cref{l:span}, $\ket{\phi_m}\in S$ for $S$ from Equation~(\ref{eqn:7}). Next, by \Cref{obs:closedform}, Hamiltonians from $F\cup P$ act invariantly on clock $C$, and Hamiltonians from $G$ can only increment $C$ from $1$ (i.e.\ its initial value in $\ket{\phi}$) to $\absA+1$. The observable $M$, however, acts non-trivially only when $C$ is set to $\absC=\absA+L+1$. The only Hamiltonians which can increment $C$ from $\absA+1$ to $\absA+L+1$ are those from $Q$. Each such $H_s\in Q$ can map $C$ from time $\absA+s$ to $\absA+s+1$ or vice versa, for $s\in\set{1,\ldots, L}$. Thus, since strictly fewer than $L$ of the $H_u$ chosen are from $Q$, it follows that $\ket{\phi_m}$ has no support on time step $\absC=\absA+L+1$, i.e.\ $(I_{AB}\otimes \ketbra{1}{1}_{C_{\absC}})\ket{\phi_m}=0$. The claim now follows since we \Cref{ass:2} says verifier $V=V_L\cdots V_1$ has its output qubit, denoted $B_1$, set to $\ket{0}$ until its final gate $V_L$ is applied.
\end{proof}

        \subsubsection{Completeness} \label{sscn:VQAcompleteness}
        With all observations and lemmas of Section \ref{sscn:VQAlemmas} in hand, we are ready to prove completeness of the construction. Specifically, in the YES case, there exists an input $y\in\set{0,1}^{\absA}$ of Hamming weight at most $g$ accepted with probability at least $2/3$ by $V$. The honest prover proceeds as follows.
        \begin{itemize}
            \item (Prepare classical proof) Prepare state (up to global phase) $\ket{\psi_0}\coloneqq  \ket{y}_A\ket{0}_B\ket{\wt{\absA+1}}_C\ket{\wto}_D$ as follows. Starting with $\ket{\phi}=\ket{0\cdots 0}_{AB}\ket{\wto}_C\ket{\wto}_D$:
                \begin{enumerate}
                    \item Set $j=1$.
                    \item If $y_j=1$ then
                    \begin{itemize}
                        \item Apply, in order, unitaries $\exp(i(\pi/2)F_1)$, $\exp(i(\pi/2)F_2)$,\ldots, $\exp(i(\pi/2)F_{\absD-1})$. This maps registers $C$ and $D$ to $1$ and $\absD$, respectively.\vspace{1mm}
                        \item Apply $\exp(i(\pi/2)P_j)$, which maps $A_j$ from $0$ to $1$.\vspace{1mm}
                        \item Apply, in order, unitaries $\exp(i(\pi/2)F_{\absD})$, $\exp(i(\pi/2)F_{\absD-1})$,\ldots, $\exp(i(\pi/2)F_1)$. This maps registers $C$ and $D$ back to $1$ and $1$, respectively.\vspace{1mm}
                    \end{itemize}
                    \item Apply unitary $\exp(i(\pi/2)G_j)$, which maps $C$ from $j$ to $j+1$.
                    \item Set $j=j+1$.
                    \item If $j<\abs{A}+1$, return to line 2 above.
                \end{enumerate}
            This process applies $g(2\absD-1)+ \absA$ gates.

            \item (Simulate verifier) Prepare the sequence of states $\ket{\psi_j}=e^{i\frac{\pi}{2} Q_j}\cdots e^{i\frac{\pi}{2} Q_1}\ket{\psi_0}$ by applying, in order, unitaries $\exp(i(\pi/2) Q_1)$, $\exp(i(\pi/2) Q_2)$,\ldots, $\exp(i(\pi/2) Q_L)$.
            Since the $j$th step of this process applies $\exp(i(\pi/2) Q_j)$, and since the state $\ket{\psi_0}$ has clock $C$ set to $\absA+1$, \Cref{obs:closedform} and \Cref{eqn:Qt} imply that
                \begin{equation}
                    e^{i\frac{\pi}{2} Q_j}\ket{\psi_{j-1}}=\left((V_j)_{R_j}\otimes\ket{\wt{\absA+j+1}}
                \,\bra{\wt{\absA+j}}_C\right) \ket{\psi_{j-1}},
            \end{equation}
            i.e.\ we increment the clock from $\absA +j$ to $\absA+j+1$ and apply the $j$th gate $V_j$. The final state obtained is thus
            $\ket{\psi_L}=\left(V_L\cdots V_1 \ket{y}_A\ket{0}_B\right) \otimes \ket{\wt{\absA+L+1}}_C\ket{\wto}_D$. This process applies $L$ gates.
        \end{itemize}
        Since $V$ accepts $y$ with probability at least $2/3$, we conclude $\bra{\psi_L}M\ket{\psi_L}\leq 1/3$, as desired. The number of Hamiltonians from $\SFGPQ$ we needed to simulate in this case is $m=g(2\absD-1)+\absA+L$, as desired.
		
        \subsubsection{Soundness} \label{sscn:VQAsoundness}
        We next show soundness. Specifically, in the NO case, for all inputs $y\in\set{0,1}^{\absA}$ of Hamming weight at most $g'$, $V$ accepts with probability at most $1/3$. So, consider any sequence of $m'=g'(2\absD-1)+\absA+L$ Hamiltonian evolutions producing state $\ket{\phi_{m'}}\coloneqq \Pi_{t=1}^{m'}e^{i\theta_u H_u}\ket{0\cdots 0}_{AB}\ket{\wto}_C\ket{\wto}_D$ for arbitrary $\theta_u\in \reals$ and Hamiltonians $H_u\in \SFGPQ$. \Cref{l:span} says we may write
        \begin{equation}\label{eqn:9}
            \ket{\phi_{m'}}=\sum_{y\in\set{0,1}^{\absA}}\sum_{s=1}^{\absC}\sum_{t=1}^{\absD}\alpha_{y,s,t} \etayst \in\Span(S)
        \end{equation}
        with $\sum_{y,s,t}\abs{\alpha_{y,s,t}}^2=1$.
        Now, for the observable \eqref{eq:def:M} follows that
        \begin{align}
            \bra{\phi_{m'}}M\ket{\phi_{m'}}&=1-\bra{\phi_{m'}}\left(\ketbra{1}{1}_{B_1}\otimes \ketbra{1}{1}_{C_{\absC}}\right)\ket{\phi_{m'}}=1-\bra{\eta}\ketbra{1}{1}_{B_1}\ket{\eta}\text{ for}\label{eqn:12}\\
            \ket{\eta}&\coloneqq \sum_{y\in\set{0,1}^{\absA}}\sum_{t=1}^{\absD}\alpha_{y,\absA+L+1,t}V_L\cdots V_1\ket{y}_A\ket{0}_B\ket{\wt{\absA+L+1}}_C\ket{\wtt}_D,
        \end{align}
        where we have used \Cref{eqn:9} and the fact that $M$ projects onto
        time step $\absC$ in register $C$.
        Now, if we applied strictly less than $L$ evolutions from $Q$, \Cref{lem:L} says we have no weight on time step $\absC$, so that $\bra{\phi_{m'}}M\ket{\phi_{m'}}=1\geq 2/3$, as required in the NO case. If, on the other hand, we applied at least $L$ evolutions from $Q$, then we must have applied at most $g'(2\absD-1)+\absA$ evolutions from $F\cup G\cup P$ (otherwise, we have a contradiction since $m'=g'(2\absD-1)+\absA+L$). \Cref{lem:HW} hence implies the right hand side of \Cref{eqn:12} equals $1-\bra{\eta_{g'}}\ketbra{1}{1}_{B_1}\ket{\eta_{g'}}$ for\footnote{Below, $\textup{HW}(y)$ denotes the Hamming weight of string $y$.}
        \begin{align}
            \ket{\eta_{g'}}\coloneqq \sum_{y\text{ s.t. }\textup{HW}(y)\leq g'}\sum_{t=1}^{\absD}\alpha_{y,\absA+L+1,t}V_L\cdots V_1\ket{y}_A\ket{0}_B\ket{\wt{\absA+L+1}}_C\ket{\wtt}_D,
        \end{align}
        where $\operatorname{HW}(y)$ denotes the Hamming weight of the bitstring $y$.
         But since any input $y$ of Hamming weight at most $g'$ is accepted with probability at most $1/3$, we conclude $\bra{\phi_{m'}}M\ket{\phi_{m'}}\geq 2/3$, as claimed.

         \subsubsection{Hardness ratio} \label{sscn:VQAratio}
         Finally, we show our reduction has the desired approximation ratio. Observe
         \begin{equation}
            \frac{m'}{m}=\frac{g'(2\absD-1)+\absA+L}{g(2\absD-1)+\absA+L}=\frac{g'(2\lceil L^{1+\delta}\rceil-1)+\absA+L}{g(2\lceil L^{1+\delta}\rceil-1)+\absA+L}.
         \end{equation}
         Since $\abs{A}\leq L$ by definition, and since we will choose $\delta>0$ as a small constant, this ratio scales asymptotically as $g'/g$.
         Recall now that
         \Cref{thm:QMSAhard} says that for any constant $\epsilon'>0$, the QMSA instance $\Pi'=(V',g,g')$ we are reducing from is QCMA-hard to approximate within $g'/g\in O((N')^{1-\epsilon'})$, for $N'$ the encoding size of $\Pi'$.
         Observe that
         \begin{align}\label{eqn:N'}
            N'\geq 2L'\log(\n'),
         \end{align}
         as $L'$ is the number of gates comprising $V'$, and each gate $V'_i$ takes $O(1)$ bits to specify (assuming a constant-size universal gate set) and $2\log \n'$ bits to indicate which pair of qubits $V'_i$ acts on.
         So it remains to compare $N'$ with the encoding size $N$ of our MIN-VQA instance $\Pi$.
         For this, observe that each Hamiltonian in $\SFGPQ$ may be specified using $O(\log(\absA+\absB+\absC+\absD))$ bits, since each $H_i$ requires $O(1)$ bits to specify the $4$-local matrix itself, and $4\log(\absA+\absB+\absC+\absD)$ bits to specify the (at most) $4$-tuple of qubits on which $H_i$ acts.
         Similarly, $M$ is specifiable using $O(\log(\absA+\absB+\absC+\absD))$ bits.
         Thus, $N\in O(\abs{\SFGPQ}\log(\absA+\absB+\absC+\absD))$, where we may bound
         \begin{align}
            \abs{\SFGPQ}&=2\absA+L+\absD-1\\
            &=2\n'+(L'+\n'+1)+(L'+\n'+1)^{1+\delta}-1\\
            &\leq2\n'+(L'+\n'+1)+(2L'+1)^{1+\delta}-1\\
            &\in O(L'^{1+\delta}),
         \end{align}
         where we have used that $\n'\in O(L')$ (otherwise, $V'$ does not have enough time to read all $\n'$ of its input bits). Recall now that to obtain the hardness ratio of our claim, we must show the following: For any desired constant $\epsilon>0$, there exist $\epsilon'>0$ and $\delta'>0$ such that
         \begin{align}\label{eqn:goal}
            \frac{g'}{g}\geq (N')^{1-\epsilon'}\geq N^{1-\epsilon}.
         \end{align}
         We know that for some $c\in \Theta(1)$,
         \begin{align}\label{eqn:bounds}
            (N')^{1-\epsilon'}\geq (2L'\log(\n'))^{1-\epsilon'}\text{ and }N\leq c(L')^{1+\delta}\log(\absA+\absB+\absC+\absD).
         \end{align}
         Equation (\ref{eqn:goal}) thus holds if
         \begin{align}\label{eqn:goal2}
            \frac{1-\epsilon'}{1+\delta}\geq \frac{(1-\epsilon)(\log c+\log\log(\absA+\absB+\absC+\absD))-(1-\epsilon')\log\log(\n^2)}{(1+\delta)\log(L')}+(1-\epsilon).
         \end{align}
         We conclude that for large enough $L'$, for any desired $\epsilon>0$, one can choose $\epsilon'>0$ and $\delta>0$ so that Equation (\ref{eqn:goal}) holds, as desired.

\end{proof}

\newcommand{\av}[1]{\left<#1\right>}


\section{Extension of the hardness results to QAOAs}
\label{scn:QAOA}


In this section, we prove \Cref{thm:QCMA_QAOA}, which is restated for convenience shortly. First, we define the optimization problem \QAOA\ covered by the theorem.

\begin{figure}
	\centering
	\includegraphics[width=120mm]{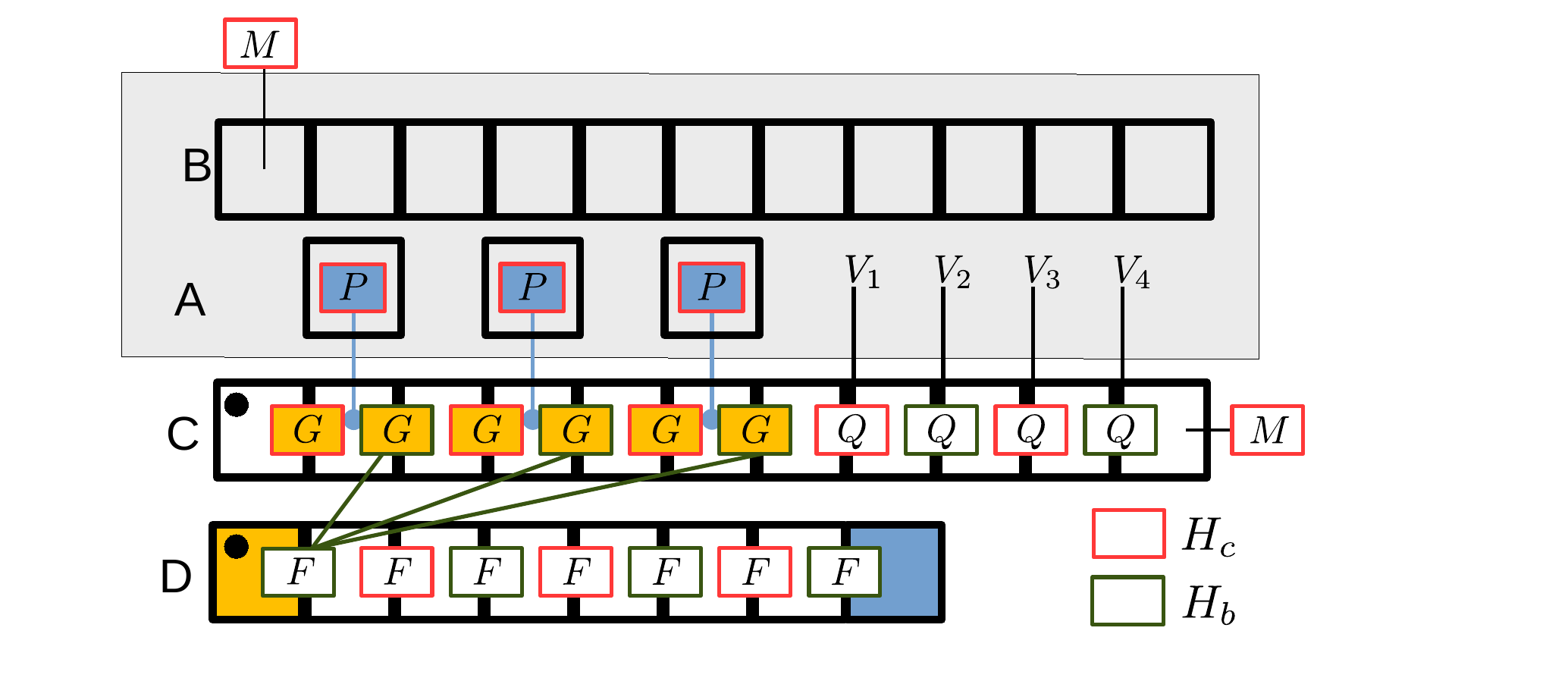}
	\caption{Figure describing the \ac{QAOA} instance (see \Cref{fig:sketch} for further details). The border color of each gate indicates if the generator belongs to $H_b$ or $H_c$. Compared to the previous VQA instance, the $P$ now only act at even time steps in $C$ and the even-indexed $G_j$ and the $F_1$ generator are combined into one generator, denoted by the red and dark green edges.
	}
	\label{fig:sketchQAOA}
\end{figure}

A $k$-local Hamiltonian is a sum of strictly $k$-local terms, i.e.\ Hermitian operators each of which acts non-trivially on at most $k$ qubtis. As mentioned previously, our definition of \QAOA\ is more general than that of~\cite{FarGolGut14}, and closer to that of~\cite{HWORVB19}.

\begin{problem}[QAOA minimization (\QAOA($k$))]\label{def:QAOA}
	For an $n$-qubit system:
	\begin{itemize}
		\item Input:
		\begin{enumerate}
			\item A set $H=\set{H_b,H_c}$ of $k$-local Hamiltonians. 
			\item A $\textup{poly}(n)$-size quantum circuit $U_b$ preparing the ground state of $H_b$, denoted $\gsb$.
			\item Integers $0\leq m\leq m'$ representing thresholds for depth.
		\end{enumerate}
		\item Output:
		\begin{enumerate}
			\item YES if there exists a sequence of angles\footnote{Throughout \Cref{def:QAOA}, for clarity we assume all angles are specified to $\poly(n)$ bits.} $(\theta_i)_{i=1}^{m}\in \RR^{m}$, such that
			\begin{align}\label{eqn:psi}
			\ket{\psi}\coloneqq e^{i\theta_{m} H_b} e^{i\theta_{m-1} H_c}\cdots e^{i\theta_2 H_b} e^{i\theta_1 H_c}\gsb
			\end{align}
			satisfies $\bra{\psi}H_c\ket{\psi}\leq \frac{1}{3}$.

			\item NO if for all sequences of angles $(\theta_i)_{i=1}^{m'}\in \RR^{m'}$
			\begin{align}
			\ket{\psi}\coloneqq e^{i\theta_{m'} H_b} e^{i\theta_{m'-1} H_c}\cdots e^{i\theta_2 H_b} e^{i\theta_1 H_c}\gsb,
			\end{align}
			satisfies $\bra{\psi}H_c\ket{\psi}\geq \frac{2}{3}$.
		\end{enumerate}
	\end{itemize}
\end{problem}
\noindent Just as for \VQA, by ``optimal depth'' of a QAOA, we mean the minimum number of Hamiltonian evolutions $m$ required above. The expectation value thresholds $\frac{1}{3}$ and $\frac{2}{3}$ are arbitrary and can be changed by rescaling and shifting $H_c$.

\thmqaoa*

\begin{proof}
    Containment in \QCMA\ is again straightforward and thus omitted. For \QCMA-hardness of approximation, we again use a reduction from an instance $\Pi=(V,g,g')$ of QMSA with $V=V_{L}\cdots V_1$ being a sequence of $L$ two-qubit gates taking in $\n$ input bits and $\m$ ancilla qubits. All those terms
    are defined as in the proof of \Cref{thm:QCMA}.

\paragraph{Proof organization.}
The proof is organized as follows.
In \cref{sec:QAOAinstance} we explain the modifications of the VAQ instances to obtain the QAOA instances of our construction. \Cref{sscn:QAOAprel} and \cref{sscn:QAOACF} explain how we recover the desired initial state and cost function.
\Cref{sec:QAOApreliminaries} provides notation preliminary technical results needed for the QCMA-completeness proof.
Then, completeness is shown in \cref{sec:QAOAcompleteness} and soundness in \cref{sec:QAOAsoundness}.
Finally, in \cref{sec:QAOAratio}, we analyze the hardness ratio achieved by the reduction.
\subsection{QCMA completeness for QAOAs}
\label{sec:QAOAinstance}
    To specify our \ac{QAOA} instance, we
    modify the set $\SFGPQ$ from the proof of \Cref{thm:QCMA} to suit our reduction here as follows. The structural changes are illustrated in Figure~\ref{fig:sketchQAOA}.
    Briefly recapping the proof techniques outline in \cref{scn:techniques}, we:
\begin{compactenum}[(i)]
    \item implement the reduction with only two generators by alternating even and odd steps of the honest prover's actions, so that $H_b$ simulates the even steps, and $H_c$ the odd ones,
    \item introduce terms $G_j$ from \cref{eqn:Ge} with $3$-cyclic behavior, i.e. allowing three logical actions,
    \item add new constraints to $H_b$ to ensure its unique ground state encodes the correct start state (see \Cref{eqn:psi} of \Cref{def:QAOA}), and
    \item add the observable $O$ to $H_c$ (scaled larger than other terms in $H_c$) to obtain the correct cost function.
\end{compactenum}
    An undesired side effect of this is that evolution by $H_c$ allows one to leave the desired logical computation space, $S$.
    We will show via \cref{lem:gate_decay} that the states obtained are still close to the set, which suffices for our soundess analysis.

    To begin, we use registers composed of $\absA=\n$, $\absB=\m$, $\absC=L+2\n+1$, and $\absD=\lceil L^{1+\delta}\rceil$ qubits, respectively, where $0<\delta<1$ is fixed by specified later.
    Without loss of generality, we assume $\absD$ and $L$ to be even.
	Additionally to the changes we outline, we also add diagonal terms additional diagonal terms. This will be relevant for defining the initial state later on.
\begin{itemize}
	\item ($F$) We remove $F_1$,
	\begin{align}
	F_j&\coloneqq \ketbra{01}{10}_{D_{j,j+1}}+\ketbra{10}{01}_{D_{j,j+1}}-2\ketbra{00}{00}_{D_{j,j+1}}  \text{ for all }j\in \set{2,\ldots, \absD-1}.\label{eqn:F2}
	\end{align}
	\item ($G$) We double the number of qubits $G$ acts on,
	\begin{align}
	G_j&\coloneqq \left(\ketbra{01}{10}_{C_{j,j+1}}+\ketbra{10}{01}_{C_{j,j+1}}\right)\otimes \ketbra{1}{1}_{D_1} -2\ketbra{001}{001}_{C_{j,j+1},D_1}
    \nonumber
    \\
	&
    \hspace{7cm}
    \text{ for all }j\in \set{1,3,\ldots, 2\absA-1},\label{eqn:Go}\\
	G_j&\coloneqq \frac{i}{\sqrt{3}}\bigl(\ketbra{0110}{1010}+\ketbra{1001}{0110}+\ketbra{1010}{1001}
    \nonumber \\ \nonumber
	&\phantom{=.}
    -\ketbra{1010}{0110}-\ketbra{0110}{1001}-\ketbra{1001}{1010}
    \bigr)_{C_{j,j+1},D_{1,2}} -2\ketbra{0010}{0010}_{C_{j,j+1},D_{1,2}} \\
	&
    \hspace{7cm}
    \text{ for all }j\in \set{2,4,\ldots, 2\absA}.\label{eqn:Ge}
	\end{align}
	While odd numbered gates can only change the clock,
    even numbered ones can
    still
    increment $C$, but also have the option of moving $\ket{\wt 1}_D\rightarrow\ket{\wt 2}_D$, which is the operation performed by $F^{(o)}_1$ in the proof of \cref{thm:QCMA} on \VQA. The superscript $(o)$ refers to the gates of the previous VQA proof.
    The following relations hold:
	\begin{align}
	e^{i\frac{\pi}{3}G_i}\ket{\wt i,\wt 1}_{C,D}&=e^{i\frac{\pi}{2}G_i^{(o)}}\ket{\wt i,\wt 1}_{C,D}\propto\ket{\wt {i+1},\wt 1}_{C,D} ,
    \\
	e^{i\frac{2\pi}{3}G_i}\ket{\wt i,\wt 1}_{C,D}&=e^{i\frac{\pi}{2}F_1^{(o)}}\ket{\wt i,\wt 1}_{C,D}\propto \ket{\wt i,\wt 2}_{C,D} ,
	\end{align}
	where, in this case, ``$\propto$'' means equality up to a phase.
	\item ($P$) For each qubit $j\in\set{1,\ldots, \absA}$ of $A$, we define the $X$-operators, but now they only act on even values in the clock register,
	\begin{align}
	P_j&\coloneqq  X_{A_j}\otimes \ketbra{1}{1}_{C_{2j}}\otimes \ketbra{1}{1}_{D_{\absD}} -2\ketbra{00}{00}_{C_{2j},D_{\absD}}\text{ for all }j\in \set{1,\ldots, \absA}.
	\end{align}
	\item ($Q$) We shift the $C$-indices of the $Q$-gates because reading in the proof takes longer time now,
	\begin{align}
	Q_{k} &\coloneqq  (V_k)_{R_k}\otimes\ketbra{01}{10}_{C_{2\absA+k,2\absA+ k+1}}+ (V_k^\dagger)_{R_k}\otimes\ketbra{10}{01}_{C_{2\absA +k,2\absA+k+1}}\\&-2\ketbra{00}{00}_{C_{2\absA+k,2\absA+k+1}} \text{ for all }k\in \set{1,\ldots, L}
	\end{align}
	\item ($M$), ($H_0$) We add the two operators
	\begin{align}
	H_0&=-\left(\sum_{i \in [\absA]}\ketbra{0}{0}_{A_i}+\sum_{i \in [\absB]}\ketbra{0}{0}_{B_i}\right) \otimes \ketbra{1}{1}_{C_1},
    \\
	M&=I-\ketbra{1}{1}_{B_1}\otimes \ketbra{1}{1}_{C_{\absC}}
	\end{align}
	to the
    set of generators.
\end{itemize}

\noindent
To construct our desired \ac{QAOA} instance, we define a partition of all gates into two groups:
\begin{align}
\mathcal{G}_1&=\{G_{i}\}_{i\in \{2,4,\dots, 2\absA\}}\cup \{F_{i}\}_{i \in \{3,5,\dots, \absD-1\}}\cup \{Q_i\}_{i\in \{2,4,\dots, L\}},
\\
\mathcal{G}_2&=\{G_i\}_{i\in \{1,3,\dots, 2\absA-1\}}\cup \{F_{i}\}_{i \in \{2,4,\dots, \absD-2\}} \cup\{Q_i\}_{i\in \{1,3,\dots, L-1\}}\cup \{P_i\}_{i\in [\absA]}.
\end{align}
Intuitively, $\mathcal{G}_1$ (respectively, $\mathcal{G}_2$) will be part of our Hamiltonian $H_b$ (respectively, $H_c$). Note also that all operators in $\mathcal{G}_1$ (respectively, $\mathcal{G}_2$) pairwise commute, a fact we will use in our analysis. Finally, in addition to \Cref{ass:1} and \Cref{ass:2} from the VQA section, we shall also use the following.
\begin{ass}\label{ass:expprec}
	The acceptance probability of $V$ in the YES (respectively, NO) case is at least $1-\epsilon_Q$ (respectively, at most $\epsilon_Q$), where $\epsilon_Q=O(N^{-1})$.
    This is achieved via standard parallel $k$ times repetition of the circuit $V$, followed by a majority vote.
    This increases the encoding size of $V$ --- for $k$ repetitions, the new gate sequence length scales with $L'=k(L+O(1))$, and yields $\epsilon_Q'=\epsilon_Q^{O(k)}$. For the precision we require, it suffices to set $k=O(\log(N))$.
\end{ass}

Due to this assumption, our encoding size increases by a multiplicative log factor, which does not affect our final approximation ratio calculation.

\subsubsection{The Min-QAOA instance}\label{sscn:QAOAinstance}
	The \ac{QAOA} instance we use to prove \Cref{thm:QCMA_QAOA} takes the generators
	\begin{align}
	H_b&=\sum_{\Gamma\in \mathcal{G}_1}\Gamma+H_0\, ,
    \\
	H_c&=\kappa\sum_{\Gamma\in \mathcal{G}_2}\Gamma+ M
	\end{align}
	with $m=g(2\absD-2)+\absC-1$ and $m'=g'(2\absD-4)+\absC-1$.
Crucially, the generators/operators comprising $H_b$ (respectively, $H_c$) pairwise commute.	
The $Q$ gates are taken from a QMSA circuit where using \Cref{ass:expprec}, we set the acceptance threshold of the circuit to
$\sqrt{\epsilon_{Q}}=\frac{1}{48m'}$.
Also, we set $\kappa=\frac{1}{24|\mathcal{G}|}$.

We proceed by first characterizing the initial state and cost function as defined in \cref{def:QAOA}.
\subsubsection{Initial state}\label{sscn:QAOAprel}
Recall that in \Cref{def:QAOA} the initial state has to be a ground state of $H_b$ (given as input via a preparation circuit $U_b$).
We want this initial state to be
\begin{align}\label{eq:gs}
\gsb=\ket{0,0,\wt 1,\wt 1}_{ABCD},
\end{align}
 which can trivially be prepared by a constant-sized circuit $U_b$.
 To see that we indeed obtain this ground state, note below that for all generators except $G_1$, $M$, which are not included in $H_b$, $\gsb$ is a ground state of the generator.
 Moreover, the groundstate turns out to be unique because
 for each qubit, the state is uniquely determined by one of the generators, which implies that the entire state is unique.
Specifically, we have
\begin{align*}
    F_i\ket{0,0,\wt 1,\wt 1}_{ABCD}&=-2\ket{0,0,\wt 1,\wt 1}_{ABCD} \quad \forall i\in\{3,5,\dots,\absD-1\},&\quad \norm{F_i}_{\infty}&=2,
    \\
    G_i\ket{0,0,\wt 1,\wt 1}_{ABCD}&=-2\ket{0,0,\wt 1,\wt 1}_{ABCD} \quad \forall i\in\{2,4\dots,2\absA\}, &\quad \norm{G_i}_{\infty}&=2,
    \\
    Q_i\ket{0,0,\wt 1,\wt 1}_{ABCD}&=-2\ket{0,0,\wt 1,\wt 1}_{ABCD} \quad \forall i\in\{2,4\dots,L\}, &\quad \norm{Q_i}_{\infty}&=2,
    \\
    H_0\ket{0,0,\wt 1,\wt 1}_{ABCD}&=-(\absA+\absB)\ket{0,0,\wt 1,\wt 1}_{ABCD}, &\quad \norm{H}_{\infty}&=\absA+\absB.
\end{align*}
Indeed, since the state \eqref{eq:gs} is a ground state of all the generators of $H_b$ and the terms of $H_b$ mutually commute, it is also a ground state of $H_b$.
Moreover, since every qubit is non-trivially supported by at least one generator of $H_b$, it is also the unique ground state for the entire Hilbert space, i.e.,
$\gsb$ represents the unique one-dimensional subspaces where each gate acts non-trivially.

\subsubsection{Cost function} \label{sscn:QAOACF}
In the \ac{QAOA} setup, the measured observable is $H_c$.
For our construction we wish to use the observable $M$.
Fortunately, we can find an upper bound for the difference of these operators.
Namely, for every $\ket{\Psi}\in \mathcal H$
\begin{align}\label{eqn:close}
|\braket{\Psi}{(H_c-M)|\Psi}|&=\kappa|\braket{\Psi}{\sum_{\Gamma\in \mathcal{G}_2}\Gamma|\Psi}|\leq 2\kappa |\mathcal{G}_2|\leq \frac{1}{12}
\end{align}
where we used (1) that $\|g\|_{\infty}\leq2$ for all $\Gamma\in\mathcal{G}_2$, and (2) the definition of $\kappa$.

\subsubsection{Preliminaries for the completeness proof}
\label{sec:QAOApreliminaries}
We first define the set of states comprising our logical computation space,
\begin{align}
S&\coloneqq \{ V_{t-2\absA-1}\cdots V_{1}\ket{y}_A\ket{0\cdots 0}_B\ket{\wt t}_C\ket{\wt s}_D |\,\forall (y,t,s)\in I_S\}
\end{align}
with
\begin{align}
I_S=\left\{(y,t,s)\bigg| y\in\set{0,1}^{\absA}, t \in\{1,\dots,\absC \},s\in
\begin{cases}
\{1,\dots,\absD\}& \text{if }t\in \{2,4,\cdots 2\absA\}\\
\{1\}& \text{otherwise}
\end{cases}\right\}
\end{align}
being the allowed index set.
Here, the notation means that $V_1$ is only applied if $t> 2\absA+1$.
Below, we often write a state $\ket{\Psi_S}\in \mathrm{span}(S)$ as
\begin{align}
\ket{\Psi_S}&=\sum_{(y,t,s)\in I_S}a_{y,t,s}V_{t-2\absA-1}\cdots V_{1}\ket{y}_A\ket{0\cdots 0}_B\ket{\wt t}_C\ket{\wt s}_D\\&\eqqcolon
\sum_{(y,t,s)\in I_S}a_{y,t,s}\psiyts\, .
\end{align}
We also define the function $W$, which is intended to capture a lower bound on the number of Hamiltonian evolutions required to prepare a given logical state $\psiyts$:
\begin{align}
W(y,t,s)&\coloneqq (2\absD-4)\mathrm{HW}(y)+t+(-1)^{\delta_{y_{\lceil t/2\rceil},1}}(s+\delta_{s,1}-2),
\end{align}
where $\mathrm{HW}(y)$ denotes the Hamming weight of $y$.

We next show a helpful lemma regarding the action of each Hamiltonian on our logical computation space, $S$.
\begin{lemma}\label{lem:nontrivial}
	The following two statements hold:
	\begin{itemize}
		\item	For every $\psiyts\in S$ and $H_i\in\{H_b,H_c\}$
	\begin{align}
	e^{\i H_i\theta}\psiyts=e^{\i\alpha^{(i)}_{y,t,s}\theta} e^{\i \Gamma^{(i)}_{y,t,s}\theta}\psiyts
	\end{align}
	for some phase $\alpha\in\reals$.
In words, applying $H_i$ simulates application of a single gate $\Gamma^{(b)}_{y,t,s}\in \mathcal G_1$, $\Gamma^{(c)}_{y,t,s}\in \mathcal \kappa \mathcal G_2\cup\{ M\}$ up to global phase $\alpha_{y,t,s}$,
where $\kappa \mathcal G_2 = \{\kappa \Gamma\mid \Gamma\in \mathcal G_2\}$.
	\item For every $\psiyts\in S$ and every gate $\Gamma\in \mathcal G_1 \cup \mathcal G_2\cup \set{H_0}$, there exist amplitudes $\set{a_{y,t,s}}$ such that
	\begin{align}
	e^{\i \Gamma\theta}\psiyts=\sum_{\substack{(y',t',s')\in I_S\\
			W(y',t',s')\leq W(y,t,s)+1}
	}a_{y,t,s}\ket{\Psi_{y',t',s'}}
	\end{align}
    In words, the application of $\Gamma$ can only increase value of the $W$-function by at most $1$.
	\end{itemize}
\end{lemma}
\begin{proof}
    The first claim will follow if
    every $\psiyts$
    is an eigenvector of all but (at most one) generator $\Gamma^{(i)}_{y,t,s}$ comprising $H_i$.
    To see why, define $H^{(i)}_{y,t,s}\coloneqq H_i-\Gamma^{(i)}_{y,t,s}$.
    Then
	\begin{align}
	e^{i\theta H_i }\psiyts&=e^{\i\theta\left(\Gamma^{(i)}_{y,t,s}+H^{(i)}_{y,t,s}\right) }\psiyts
	=e^{\i\theta \Gamma^{(i)}_{y,t,s}}e^{\i\theta H^{(i)}_{y,t,s}} \psiyts
	=e^{\i\theta \Gamma^{(i)}_{y,t,s}}e^{\i \theta \alpha^{(i)}_{y,t,s}}\psiyts,
	\end{align}
	where $ \alpha^{(i)}_{y,t,s}$ is the corresponding eigenvalue of $H^{(i)}_{y,t,s}\coloneqq H_i-\Gamma^{(i)}_{y,t,s}$. The second step uses that all generators in $H_i$ commute with each other restricted to states where $C$ and $D$ are in valid logical time states, which one can verify by direct calculation.

 The second claim of the lemma follows if every generator maps only between states $\psiyts$ and $\ket{\Psi_{y',t',s'}}$ with $|W(y,t,s)-W(y',t',s')|\leq1$.
To obtain these claims, we first list all non-trivial generator transitions of $H_b$, where one can transition between various states $\psiyts$ listed below. For example, in \Cref{eqn:84} and \Cref{eqn:85}, $F_i$ can map $\ket{y}_A\ket{0}_B\ket{\wt{2j}}_C\ket{\wt i}_D$ to $\ket{y}_A\ket{0}_B\ket{\wt{2j}}_C\ket{\wt{i+1}}_D$ and vice versa. To the right of each of these states, we list the value of the $W$-function for that state (which, recall, is only a function of $(y,t,s)$ in $\ket{y}_A\ket{0}_B\ket{\wt t}_C\ket{\wt s}_D$):
	\begin{itemize}
		\item ($F_i$, $ i\in\{3,5\dots,\absD-1\}$), $\forall j\in [\absA]$, $y \in \{0,1 \}^{\absA}$
		\begin{align}
		\ket{y}_A\ket{0}_B\ket{\wt{2j}}_C\ket{\wt i}_D:\quad& W=(2\absD-4)\mathrm{HW}(y)+2j+(-1)^{\delta_{y_{\lceil t/2\rceil},1}}(i-2)\label{eqn:84}\\
		\ket{y}_A\ket{0}_B\ket{\wt{2j}}_C\ket{\wt{i+1}}_D:\quad& W=(2\absD-4)\mathrm{HW}(y)+2j+(-1)^{\delta_{y_{\lceil t/2\rceil},1}}(i-1)\label{eqn:85}
		\end{align}
		\item ($G_i$, $i\in\{2,4\dots,2\absA\}$), $\forall y \in \{0,1 \}^{\absA}$
		\begin{align}
		\ket{y}_A\ket{0}_B\ket{\wt i}_C\ket{\wt 1}_D:\quad& W=(2\absD-4)\mathrm{HW}(y)+i\\
		\ket{y}_A\ket{0}_B\ket{\wt{i+1}}_C\ket{\wt 1}_D:\quad& W=(2\absD-4)\mathrm{HW}(y)+i+1\\
		\ket{y}_A\ket{0}_B\ket{\wt i}_C\ket{\wt 2}_D:\quad& W=(2\absD-4)\mathrm{HW}(y)+i
		\end{align}
		\item ($Q_i$, $i\in\{2,4,\dots,L\}$), $\forall y \in \{0,1 \}^{\absA}$
		\begin{align}
		V_{i-1}\cdots V_1\ket{y}_A\ket{0}_B\ket{\wt{2\absA+1+i}}_C\ket{\wt 1}_D:\quad& W=(2\absD-4)\mathrm{HW}(y)+2\absA+1+i\\
		V_{i}\cdots V_1\ket{y}_A\ket{0}_B\ket{\wt{2\absA+2+i}}_C\ket{\wt 1}_D:\quad& W=(2\absD-4)\mathrm{HW}(y)+2\absA+2+i
		\end{align}
		We note that, by \Cref{ass:1}, since $V_i$ can only be controlled via register $A$ (as opposed to acting on $A$ as a target register), it cannot change the string $y$ in $A$. Thus, $W$ is only affected via the change on the $C$ register.
		\item ($H_0$), $\forall y \in \{0,1 \}^{\absA}$
		\begin{align}
		\ket{y}_A\ket{0}_B\ket{\wt 1}_C\ket{\wt 1}_D:\quad& W=(2\absD-4)\mathrm{HW}(y)+1
		\end{align}
	\end{itemize}

	In all cases above, the change in $W$ is at most $1$, every logical state $\psiyts$ in $S$ is appears in precisely one set, and $H_0$ always acts trivially, proving the lemma for $H_b$. Repeating this approach for $H_c$ yields a similar conclusion:
	\begin{itemize}
		\item ($F_i$, $i\in\{2,4\dots,\absD-2\}$), $\forall j\in [\absA]$, $y \in \{0,1 \}^{\absA}$
		\begin{align}
		\ket{y}_A\ket{0}_B\ket{\wt{2j}}_C\ket{\wt i}_D:\quad& W=(2\absD-4)\mathrm{HW}(y)+2j+(-1)^{\delta_{y_{\lceil t/2\rceil},1}}(i-2)\\
		\ket{y}_A\ket{0}_B\ket{\wt{2j}}_C\ket{\wt{i+1}}_D:\quad& W=(2\absD-4)\mathrm{HW}(y)+2j+(-1)^{\delta_{y_{\lceil t/2\rceil},1}}(i-1)
		\end{align}
		\item ($G_i$, $i\in\{1,3\dots,2\absA-1\}$), $\forall y \in \{0,1 \}^{\absA}$
		\begin{align}
		\ket{y}_A\ket{0}_B\ket{\wt i}_C\ket{\wt 1}_D:\quad& W=(2\absD-4)\mathrm{HW}(y)+i\\
		\ket{y}_A\ket{0}_B\ket{\wt{i+1}}_C\ket{\wt 1}_D:\quad& W=(2\absD-4)\mathrm{HW}(y)+i+1
		\end{align}
		\item ($P_i$, $i\in\{1,\dots,\absA\}$), $\forall y \in \{\{0,1 \}^{\absA}|y_i=0\}$
		\begin{align}
		\ket{y}_A\ket{0}_B\ket{\wt{2i}}_C\ket{\wt \absD}_D:\quad& W=(2\absD-4)\mathrm{HW}(y)+2i+(\absD-2)\\
		\ket{y\oplus e_i}_A\ket{0}_B\ket{\wt{2i}}_C\ket{\wt \absD}_D:\quad& W=(2\absD-4)(\mathrm{HW}(y)+1)+2i-(\absD-2)
		\end{align}
		\item ($Q_i$, $i\in\{1,3,\dots,L-1\}$), $\forall y \in \{0,1 \}^{\absA}$
		\begin{align}
		V_{i-1}\cdots V_1\ket{y}_A\ket{0}_B\ket{\wt{2\absA+1+i}}_C\ket{\wt 1}_D:\quad& W=(2\absD-4)\mathrm{HW}(y)+2\absA+1+i\\
		V_{i}\cdots V_1\ket{y}_A\ket{0}_B\ket{\wt{2\absA+2+i}}_C\ket{\wt 1}_D:\quad& W=(2\absD-4)\mathrm{HW}(y)+2\absA+2+i
		\end{align}
		\item ($M$), $\forall y \in \{0,1 \}^{\absA}$
		\begin{align}
		V_{L}\cdots V_1\ket{y}_A\ket{0}_B\ket{\wt \absC}_C\ket{\wt 1}_D:\quad& W=(2\absD-4)\mathrm{HW}(y)+\absC\\
		Z_{B_1}V_{L}\cdots V_1\ket{y}_A\ket{0}_B\ket{\wt \absC}_C\ket{\wt 1}_D:\quad& W=(2\absD-4)\mathrm{HW}(y)+\absC
		\end{align}
	\end{itemize}
Note that the first claim of the lemma indeed holds for $M$, but since $Z_{B_1}V_{L}\cdots V_1\ket{y}_A\ket{0}_B\ket{\wt \absC}_C\ket{\wt 1}_D\not\in \mathrm{span}(S)$, the second claim does not (and thus why we write $\Gamma\in \mathcal G_1 \cup \mathcal G_2\cup \set{H_0}$ in the second claim).
\end{proof}


\subsubsection{Completeness}
\label{sec:QAOAcompleteness}
In the YES case, there exists a sequence of gates with proof $y\in\set{0,1}^{\absA}$ of Hamming weight at most $g$ accepted with probability at least $1-\epsilon_{Q}$ by the verifier circuit $V$.
We use shorthand notation $(y)_{j}=(y_1,\dots y_{j-1},0,\dots,0)$ to indicate the partially written proof string.
Also, $\exp(\i\theta H_i)\sim \exp(\i\theta \Gamma)$ indicates which generator $\Gamma$ in $H_i$ performs the non-trivial operation (as per \Cref{lem:nontrivial}, claim 1). The honest prover proceeds as follows:
\begin{itemize}
	\item (Prepare classical proof) Prepare state (up to global phase)
    $\ket{\psi_0}\coloneqq  \ket{y}_A\ket{0}_B\ket{\wt{2\absA+1}}_C\ket{\wto}_D$ as follows. Starting with $\gsb=\ket{(y)_0,0,\wt 1,\wt 1}_{ABCD}$:
	\begin{enumerate}
		\item Set $j=1$.
		\item  Apply $\exp(\i\frac{\pi}{2\kappa}H_c)\sim \exp(\i\frac{\pi}{2}G_{2j-1})$ to map $\ket{\wt {2j-1}}_C\rightarrow \ket{\wt {2j}}_C$. This maps
\begin{align}
    \ket{(y)_{j-1},0,\wt{2j-1},\wt{1}}_{ABCD} \quad \mapsto \quad \ket{(y)_{j-1},0,\wt{2j},\wt{1}}_{ABCD}.
\end{align}
		\item If $y_j=1$ then
		\begin{itemize}
			\item Apply $\exp(\i\frac{2\pi}{3}H_b)\sim \exp(\i\frac{2\pi}{3}G_{2j})$, to map $\ket{\wt {1}}_D\rightarrow \ket{\wt 2}_D$, i.e.
\begin{align}
    \ket{(y)_{j-1},0,\wt{2j},\wt{1}}_{ABCD} \quad \mapsto \quad \ket{(y)_{j-1},0,\wt{2j},\wt{2}}_{ABCD}.
\end{align}
			\item Apply, in order, $\exp(\i\frac{\pi}{2\kappa}H_c)\sim \exp(\i\frac{\pi}{2}F_{2})$, $\exp(\i\frac{\pi}{2}H_b)\sim \exp(\i\frac{\pi}{2}F_{3})$, \ldots, $\exp(\i\frac{\pi}{2}H_b)\sim \exp(\i\frac{\pi}{2}F_{\absD-1})$, in total $\absD-2$ operations . This maps $\ket{\wt {2}}_D\rightarrow \ket{\wt {\absD}}_D$, i.e.
\begin{align}
    \ket{(y)_{j-1},0,\wt{2j},\wt{2}}_{ABCD} \quad \mapsto \quad \ket{(y)_{j-1},0,\wt{2j},\wt{\absD}}_{ABCD}.
\end{align}
			\item Apply $\exp(\i\frac{\pi}{2\kappa}H_c)\sim \exp(\i\frac{\pi}{2}P_j)$, to map $\ket{0}_{A_{j}}\rightarrow \ket{1}_{A_{j}}$, i.e.
\begin{align}
    \ket{(y)_{j-1},0,\wt{2j},\wt{\absD}}_{ABCD} \quad \rightarrow \quad \ket{(y)_{j},0,\wt{2j},\wt{\absD}}_{ABCD}.
\end{align}
			\item 			Apply, in order, $\exp(\i\frac{\pi}{2}H_b)\sim \exp(\i\frac{\pi}{2}F_{\absD-1})$,$\exp(\i\frac{\pi}{2\kappa}H_c)\sim \exp(\i\frac{\pi}{2}F_{\absD-2})$ \ldots, $\exp(\i\frac{\pi}{2\kappa}H_c)\sim \exp(\i\frac{\pi}{2}F_{2})$, in total $\absD-2$ operations. This maps $\ket{\wt {\absD}}_D\rightarrow \ket{\wt {2}}_D$, i.e.
\begin{align}
    \ket{(y)_{j},0,\wt{2j},\wt{\absD}}_{ABCD} \quad \mapsto \quad \ket{(y)_{j},0,\wt{2j},\wt{2}}_{ABCD}.
\end{align}
			\item 			 Apply $\exp(\i\frac{2\pi}{3}H_b)\sim \exp(\i\frac{2\pi}{3}G_{2j})$, to map $\ket{\wt {2}}_D\rightarrow \ket{\wt 1}_D$  and $\ket{\wt {2j}}_C\rightarrow \ket{\wt {2j+1}}_C$, i.e.
\begin{align}
    \ket{(y)_{j},0,\wt{2j},\wt{2}}_{ABCD} \quad \mapsto \quad \ket{(y)_{j},0,\wt{2j+1},\wt{1}}_{ABCD}
\end{align}
		\end{itemize}
		\item else
		\begin{itemize}
			\item 			Apply $\exp(\i\frac{\pi}{3}H_b)\sim \exp(\i\frac{\pi}{3}G_{2j})$, to map $\ket{\wt {2j}}_C\mapsto \ket{\wt {2j+1}}_C$
, i.e.
\begin{align}
    \ket{(y)_{j-1},0,\wt{2j},\wt{1}}_{ABCD} \quad \mapsto \quad \ket{(y)_{j},0,\wt{2j+1},\wt{1}}_{ABCD}.
\end{align}
		\end{itemize}
		\item Set $j=j+1$.
		\item If $j<\abs{A}$, return to line 2 above.
	\end{enumerate}
	
	This process applies $2g(\absD-1)+ 2\absA$ gates.
	
	\item (Simulate verifier) Apply in order, $\exp(\i\frac{\pi}{2\kappa}H_c)\sim\exp(\i\frac{\pi}{2}Q_{1})$ , $\exp(\i\frac{\pi}{2}H_b)\sim\exp(\i\frac{\pi}{2}Q_{2})$, \ldots,
    $\exp(\i\frac{\pi}{2}H_b)\sim\exp(\i\frac{\pi}{2}Q_{L})$ for a total $L$ gates.
    This implements the verification circuit, i.e.\
\begin{align}
    \ket{y,0,\wt{2\absA+1},\wt{1}}_{ABCD} \quad \rightarrow \quad\ket{\Psi_L}\coloneqq V_L\cdots V_1\ket{y,0,\wt{\absC},\wt{1}}_{ABCD}.
\end{align}
\end{itemize}
Since $V$ accepts proof $y$ of the QMSA instance with probability at least $1-\epsilon_Q$, we conclude using \Cref{eqn:close} that
\begin{align}
\bra{\Psi_L}H_c\ket{\psi_L}\leq \bra{\Psi_L}M\ket{\Psi_L}+\frac{1}{12} \leq 1-(1-\epsilon_Q)+\frac{1}{12}\leq \frac{1}{3}
\end{align}
as desired.
The number of Hamiltonians applied in this case is $m=g(2\absD-2)+2\absA+L=g(2\absD-2)+\absC-1$, also as desired.

\subsubsection{Soundness}
\label{sec:QAOAsoundness}
In the proof of \Cref{thm:QCMA} for \VQA, we showed that all time evolutions with the constructed generators keep us in our desired logical computation space $S$.
In contrast, for our \QAOA\ construction, the $M$ operator (embedded in $H_c$) does \emph{not} necessarily preserve the space $\textup{span}(S)$ (see Claim 2 of \Cref{lem:nontrivial}). We thus first require the following lemma, which allows us to ``round'' our intermediate state back to one in $S$ for our analysis and also establishes $W(y,t,s)$ as a proper lower bound for the number of gate applications required to reach the states in $S$.

 \begin{lemma}[Rounding lemma]\label{lem:gate_decay}
    We fix $\epsilon_Q$ as in \cref{ass:expprec}.
	In the NO case, after $\zeta\geq 1$ applications of $H_c$ and $H_b$, the state
	\begin{align}
	\ket{\Psi_\zeta}\in\Gamma_\zeta\coloneqq \left\{\prod_{i=1}^\zeta e^{i H_i \theta_i}\gsb \mid H_i\in\{H_b,H_c\}, \theta\in \RR^\zeta\right\}
	\end{align}
	 is $\epsilon\leq 4\zeta \sqrt{\epsilon_Q}$ close to the span of $S$, i.e.,
	\begin{align}\label{eqn:norm}
	\forall \ket{\Psi_\zeta}\in \Gamma_\zeta,~\exists \ket{\Psi'_\zeta}\in \mathrm{span}(S):
	\norm{\ketbra{\Psi_\zeta}{\Psi_\zeta}-\ketbra{\Psi'_\zeta}{\Psi'_\zeta}}_{\mathrm{tr}}\leq 4 \zeta \sqrt{\epsilon_Q}
	\end{align}
	and it additionally holds that
	\begin{align}
	\ket{\Psi'_{\zeta}}=
	\sum_{\substack{(y,t,s)\in I_S\\
			W(y,t,s)\leq \zeta+1}
	}a_{y,t,s}\psiyts.
	\end{align}
\end{lemma}
This lemma is needed because the time evolution of the observable $M$ (in $H_c$) may leave the sub-space $\Span(S)$.
The rounding step is possible, because in the NO case, the state in the $B_1$ register, after applying the circuit $V$ ($\tilde s=|D|$), is always close to $\ket{0}_{B_1}$ (using \cref{ass:expprec}). This means that in NO instances, the evolution of $M$ only adds to a global phase. 
\begin{proof}
	For our construction we use
	\begin{align}
	\ket{\Psi'_{\zeta+1}}=e^{\i \theta_{\zeta+1} (H_{\zeta+1}-M\delta_{H_{\zeta+1},H_c})}\ket{\Psi'_{\zeta}},
	\end{align}
	i.e.\ a similar ansatz as for $\ket{\Psi_\zeta}$ but without the $M$ generator in $H_c$. 
		We prove the lemma by induction.
        The lemma statement holds trivially for the base case $\zeta=0$, since
		\begin{align}
			\ket{\Psi_{\zeta=0}}&=\ket{\Psi'_{\zeta=0}}=\ket{0,0,\wt 1,\wt 1}_{ABCD}
			=\sum_{\substack{(y,t,s)\in I_S\\
					W(y,t,s)\leq 1}
			}a_{y,t,s}\psiyts
		\end{align}
		with $W(0,1,1)=1$.
		
		Induction step: For the norm inequality \eqref{eqn:norm}, only $M$ maps states outside of $S$, meaning we only have to consider the action of $H_c$. Then,
		\begin{align*}
		\norm{\ketbra{\Psi_{\zeta+1}}{\Psi_{\zeta+1}}-\ketbra{\Psi'_{\zeta+1}}{\Psi_{\zeta+1}'}}_\mathrm{tr}
		&= \norm{e^{\i H_c\theta}\ketbra{\Psi_{\zeta}}{\Psi_{\zeta}}e^{-\i H_c\theta}-e^{\i(H_c-M)\theta}\ketbra{\Psi'_{\zeta}}{\Psi_{\zeta}'}e^{-\i(H_c-M)\theta}}_\mathrm{tr}\\
		&= \norm{\ketbra{\Psi_{\zeta}}{\Psi_{\zeta}}-e^{-\i M\theta}\ketbra{\Psi'_{\zeta}}{\Psi_{\zeta}'}e^{\i M\theta}}_\mathrm{tr}\\
		&\leq \norm{\ketbra{\Psi'_{\zeta}}{\Psi'_{\zeta}}-e^{-\i M\theta}\ketbra{\Psi'_{\zeta}}{\Psi_{\zeta}'}e^{\i M\theta}}_\mathrm{tr} +\norm{\ketbra{\Psi'_{\zeta}}{\Psi_{\zeta}'}-\ketbra{\Psi_{\zeta}}{\Psi_{\zeta}}}_\mathrm{tr}\\
		&\leq \norm{\ketbra{\Psi'_{\zeta}}{\Psi'_{\zeta}}-e^{-\i M\theta}\ketbra{\Psi'_{\zeta}}{\Psi_{\zeta}'}e^{\i M\theta}}_\mathrm{tr}+4\zeta \epsilon_{Q},
		\end{align*}
		where the second statement holds since $[H_c,M]=0$ and by the unitary invariance of the trace norm, the third by the triangle inequality, and the fourth by the induction hypothesis. Now, for $M$ we have the following non-trivial action:
		\begin{align}
		e^{-i\theta}e^{iM\theta}\ket{\Psi'_\zeta}-\ket{\Psi'_\zeta}&=\sum_{(y,t,s)\in I_S}a_{y,t,s}(e^{\i (M-1)\theta }-1)\psiyts\\
		&=
		\sum_{y\in\{0,1\}^{\absA}}a_{y,\absC,1}(e^{-\i\theta\ketbra{1}{1}_{B_1}}-1)\ket{\Psi_{y,\absC,1}}\\
		&=(e^{-i\theta}-1)\sum_{y\in\{0,1\}^{\absA}}a_{y,\absC,1}\ketbra{1}{1}_{B_1}\ket{\Psi_{y,\absC,1}},
		\end{align}
		where we used that $M$ only acts non-trivially for $t=\absC$.
		This means
		\begin{align}
		\norm{\ketbra{\Psi'_{\zeta}}{\Psi'_{\zeta}}-e^{-\i M\theta}\ketbra{\Psi'_{\zeta}}{\Psi_{\zeta}'}e^{\i M\theta}}_\mathrm{tr}&\leq 2
		\norm{e^{-\i\theta}e^{iM\theta}\ket{\Psi'_\zeta}-\ket{\Psi'_\zeta}}_2\\
		&\leq 4\sqrt{ \sum_{y\in\{0,1\}^{\absA}}|a_{y,\absC,1}|^2\bra{\Psi_{y,\absC,1}}\ketbra{1}{1}_{B_1}\ket{\Psi_{y,\absC,1}}}\\
		&\leq 4 \sqrt{\sum_{y\in\{0,1\}^{\absA}}|a_{y,\absC,1}|^2 \epsilon_Q}\leq 4\sqrt{\epsilon_Q},
		\end{align}
		where the first line is a known norm inequality\footnote{%
        due to
        $\norm{\ketbra\psi\psi - \ketbra\phi\phi}_{\tr} = 2\sqrt{1-\abs{\braket\psi\phi}^2}$
        and
        $\norm{\ket\psi-\ket\phi}_2 = \sqrt{2-2\Re(\braket{\psi}{\phi})}$ for all $\ket\psi,\ket\phi \in \H$
        and $\sqrt{1-x^2} \leq \sqrt{2-2 x} \quad \forall x\in [0,1]$.
        }
        and we used that $\bra{\Psi_{y,\absC,1}}\ketbra{1}{1}_{B_1}\ket{\Psi_{y,\absC,1}}\leq\epsilon_Q$ is the acceptance probability of a QMSA NO instance. This shows the first claim of the lemma. For the second claim, a similar induction setup, coupled with \Cref{lem:nontrivial}, yields

		\begin{align}
		\ket{\Psi'_{\zeta+1}}&=e^{\i \theta_{\zeta+1} (H_{\zeta+1}-M\delta_{H_{\zeta+1},H_c})}\ket{\Psi'_{\zeta}}\\
		&=\sum_{\substack{(y,t,s)\in I_S\\
				W(y,t,s)\leq \zeta+1}
		}a_{y,t,s} e^{\i \theta_{\zeta+1}(\Gamma^{({\zeta+1})}_{y,t,s}+\alpha^{({\zeta+1})}_{y,t,s})}\psiyts\\
		&=\sum_{\substack{(y,t,s)\in I_S\\
				W(y,t,s)\leq \zeta+1}
		}a_{y,t,s}\sum_{\substack{(y',t',s')\in I_S\\
				W(y',t',s')\leq W(y,t,s)+1}
		}b_{y',t',s'}^{(y,t,s)}\ket{\Psi_{y',t',s'}}\\
		&=\sum_{\substack{(y,t,s)\in I_S\\
				W(y,t,s)\leq \zeta+2}}
		a'_{y,t,s}\psiyts,
		\end{align}
		where the second and third statements follow from the first and second claims of \Cref{lem:nontrivial}, respectively, and the last statement just recombines the $a$ and $b$ indices into new indices $a'$.
\end{proof}
We are finally ready to prove soundness. For this, we need to show that in the NO case, all sequences of $\zeta\leq m'=g'(2\absD-4)+\absC-1$ gates produce  cost function value $\braket{\Psi_\zeta}{H_c|\Psi_\zeta}\geq\frac{2}{3}$.
This follows since for all $\zeta\leq m'$,
\begin{align}
	\bra{\Psi_{\zeta}}H_c\ket{\Psi_\zeta}
	&\geq \bra{\Psi_{\zeta}}M\ket{\Psi_{\zeta}}-\frac{1}{12}\\
	&\geq \bra{\Psi'_{\zeta}}M\ket{\Psi'_{\zeta}} -|\mathrm{Tr}[M(\ketbra{\Psi_\zeta}{\Psi_\zeta}-\ketbra{\Psi'_\zeta}{\Psi'_\zeta})]|-\frac{1}{12}\\
	&\geq\bra{\Psi'_{\zeta}}M\ket{\Psi'_{\zeta}} -\norm{M}_{\infty}\norm{\ketbra{\Psi_\zeta}{\Psi_\zeta}-\ketbra{\Psi'_\zeta}{\Psi'_\zeta}}_{\mathrm{tr}}-\frac{1}{12}\\
	&\geq \bra{\Psi'_{\zeta}}M\ket{\Psi'_{\zeta}} -4m'\sqrt{\epsilon_{Q}}-\frac{1}{12}\\
	&\geq \bra{\Psi'_{\zeta}}M\ket{\Psi'_{\zeta}}-\frac{1}{6}
\end{align}
where the first statement follows from \Cref{eqn:close}, the third by H\"{o}lder's inequality, the fourth by \Cref{lem:gate_decay}, and the last since $\sqrt{\epsilon_Q}\leq \frac{1}{48 m'}$.
By \Cref{lem:nontrivial}, we can expand $\ket{\Psi'_{\zeta}}$ in the basis
\begin{align}
\ket{\Psi'_{\zeta}}=
\sum_{\substack{(y,t,s)\in I_S\\
W(y,t,s)\leq m'+1}
}a_{y,t,s}\psiyts
\end{align}
 which gives
\begin{align}
\bra{\Psi'_{\zeta}}M\ket{\Psi'_{\zeta}}
=1- \sum_{\substack{y\in\{0,1\}^{\absA} |\mathrm{HW}(y)\leq g'\\
}
}|a_{y,\absC,1}|^2\bra{\Psi_{y,\absC,1}}\ketbra{1}{1}_{B_1}\ket{\Psi_{y,\absC,1}}\geq 1-\epsilon_Q
\end{align}
as $M$ only acts non-trivial on $t=\absC$ and $W(y,\absC,1)\leq m'+1$ reduces to $\mathrm{HW}(y)\leq g'$, and in the NO case QMSA accepts such a $y$ with at most $\epsilon_Q$ probability.
Combining the two results we get
\begin{align}
\bra{\Psi_{\zeta}}H_c\ket{\Psi_\zeta}\geq 1-\epsilon_Q-\frac16> \frac{2}{3}
\end{align}
 which shows soundness for all gates-sequences of length $\zeta\leq m'$.

\subsubsection{Hardness ratio}
\label{sec:QAOAratio}
The analysis is essentially identical to that for \VQA, so we sketch it briefly. Since we set $\absD=\lceil L^{1+\delta}\rceil$, we have that in ratio

\begin{align}
\frac{m'}{m}=\frac{g'(2\absD-4)+\absC-1}{g(2\absD-2)+\absC-1},
\end{align}
the dominant term is again $\absD$. Thus, $m'/m\approx g'/g\geq O((N')^{1-\epsilon'})$, for $N'$ the encoding size of the QMSA instance, and for any desired $\epsilon'>0$. Since the encoding size of $H_b$ and $H_c$ can also be seen to scale as $O((L')^{1+\delta})$ (recall $L'$ the number of gates in the original QMSA circuit $V'$), we can apply \Cref{eqn:bounds} and the surrounding approximation ratio analysis from \VQA\ to argue again that $N\approx O((N')^{1+\epsilon'})$, for $N$ the encoding size of our \QAOA\ instance with logarithmic overhead to satisfy \cref{ass:expprec} and only $O(L)$ overhead for the changes performed to the gate set.
Thus, for any desired $\epsilon>0$, we may choose $\epsilon'>0$ and $\delta>0$ so that $g'/g\geq (N')^{1-\epsilon'}\geq N^{1-\epsilon}$, as desired.

\end{proof}

\section*{Acknowledgements}
We thank Ashley Montanaro for helpful discussions. 
LB and MK are supported by the German Federal Ministry of Education and Research (BMBF) within
the funding program ``Quantum Technologies -- from Basic Research to Market'' via the joint project MANIQU
(grant number 13N15578) and the Deutsche Forschungsgemeinschaft (DFG, German Research Foundation) under the grant number 441423094 within the Emmy Noether Program.
SG was supported by the DFG under grant numbers 450041824 and 432788384, the BMBF within the funding program ``Quantum Technologies -- from Basic Research to Market'' via project PhoQuant (grant number 13N16103), and the project ``PhoQC'' from the programme ``Profilbildung 2020'', an initiative of the Ministry of Culture and Science of the State of Northrhine Westphalia. The sole responsibility for the content of this publication lies with the authors.

\begin{acronym}[POVM]
\acro{ACES}{averaged circuit eigenvalue sampling}
\acro{AGF}{average gate fidelity}
\acro{AP}{Arbeitspaket}

\acro{BOG}{binned outcome generation}

\acro{CNF}{conjunctive normal form}
\acro{CP}{completely positive}
\acro{CPT}{completely positive and trace preserving}
\acro{cs}{computer science}
\acro{CS}{compressed sensing} 
\acro{ctrl-VQE}{ctrl-VQE}

\acro{DAQC}{digital-analog quantum computing}
\acro{DD}{dynamical decoupling}
\acro{DFE}{direct fidelity estimation} 
\acro{DFT}{discrete Fourier transform}
\acro{DM}{dark matter}

\acro{FFT}{fast Fourier transform}

\acro{GST}{gate set tomography}
\acro{GTM}{gate-independent, time-stationary, Markovian}
\acro{GUE}{Gaussian unitary ensemble}

\acro{HOG}{heavy outcome generation}

\acro{irrep}{irreducible representation}

\acro{LDPC}{low density partity check}
\acro{LP}{linear program}

\acro{MAGIC}{magnetic gradient induced coupling}
\acro{MBL}{many-body localization}
\acro{MIP}{mixed integer program}
\acro{ML}{machine learning}
\acro{MLE}{maximum likelihood estimation}
\acro{MPO}{matrix product operator}
\acro{MPS}{matrix product state}
\acro{MS}{M{\o}lmer-S{\o}rensen}
\acro{MSE}{mean squared error}
\acro{MUBs}{mutually unbiased bases} 
\acro{mw}{micro wave}

\acro{NISQ}{noisy and intermediate scale quantum}

\acro{ONB}{orthonormal basis}
\acroplural{ONB}[ONBs]{orthonormal bases}

\acro{POVM}{positive operator valued measure}
\acro{PQC}{parametrized quantum circuit}
\acro{PSD}{positive-semidefinite}
\acro{PSR}{parameter shift rule}
\acro{PVM}{projector-valued measure}

\acro{QAOA}{Quantum Approximate Optimization Algorithm}
\acro{QC}{quantum computation}
\acro{QEC}{quantum error correction}
\acro{QFT}{quantum Fourier transform}
\acro{QM}{quantum mechanics}
\acro{QML}{quantum machine learning}
\acro{QMT}{measurement tomography}
\acro{QPT}{quantum process tomography}
\acro{QPU}{quantum processing unit}
\acro{QUBO}{quadratic binary optimization}

\acro{RB}{randomized benchmarking}
\acro{RBM}{restricted Boltzmann machine}
\acro{RDM}{reduced density matrix}
\acro{rf}{radio frequency}
\acro{RIC}{restricted isometry constant}
\acro{RIP}{restricted isometry property}
\acro{RMSE}{root mean squared error}

\acro{SDP}{semidefinite program}
\acro{SFE}{shadow fidelity estimation}
\acro{SIC}{symmetric, informationally complete}
\acro{SPAM}{state preparation and measurement}
\acro{SPSA}{simultaneous perturbation stochastic approximation}

\acro{TT}{tensor train}
\acro{TM}{Turing machine}
\acro{TV}{total variation}

\acro{VQA}{variational quantum algorithm}

\acro{VQE}{variational quantum eigensolver}

\acro{XEB}{cross-entropy benchmarking}

\end{acronym}

\addcontentsline{toc}{section}{References}
\printbibliography

\appendix

\section{Additional proofs}\label{app:A}

\begin{proof}[Proof of \Cref{cor:consequence}]
    Suppose there exists an algorithm $A$ for computing estimate $m_{\textup{est}} \in [m_{\opt},\allowbreak N^{1-\epsilon}m_{\opt}]$.
    We show how to use $A$ to decide \VQA, yielding QCMA-hardness.
    Specifically, given an instance $\Pi$ of \VQA, run $A$.
    If $A$'s output is less than or equal to $m'$, accept.
    Otherwise, reject.

    To see that this is correct, observe that in the YES case, $m_{\opt}\leq m$.
    Since $m'/m\geq N^{1-\epsilon}$, $A$ outputs estimate $m_{\textup{est}}\leq m'$, from which we conclude $\Pi$ is cannot be a NO instance, and thus must be a YES instance (due to the promise that $\Pi$ is either a YES or NO instance).
    Conversely, in the NO case, $m_{\textup{est}}\geq m_{\opt}>m'$, from which we conclude $\Pi$ is a NO instance.
\end{proof}
\end{document}